\documentclass[a4paper,11pt]{article}
\usepackage[margin=1.2in]{geometry}
\pdfoutput=1
\usepackage{mathtools}

\usepackage[utf8]{inputenc}
\usepackage{amsfonts,amssymb, bm,amsthm}
\usepackage{graphicx,color,epstopdf}
\usepackage[colorlinks,
            linkcolor=red,
            anchorcolor=blue,
            citecolor=green
            ]{hyperref}
\usepackage[inline]{showlabels}
\allowdisplaybreaks
\newtheorem{myprop}{Proposition}[section]
\newtheorem{mylemma}{Lemma}[section]
\newtheorem{mytheorem}{Theorem}[section]

\newtheorem{myremark}{Remark}[section]

\def\XXint#1#2#3{{\setbox0=\hbox{$#1{#2#3}{\int}$}
    \vcenter{\hbox{$#2#3$}}\kern-.5\wd0}}
\def\dd{''}

\def\Op {{\rm Op}}
\def\bb {{\bf b}}
\def\bc {{\bf c}}

\def\bR {{\bf R}}
\def\bF {{\bf F}}
\def\bG {{\bf G}}
\def\bE {{\bf E}}
\def\bI {{\bf I}}
\def\bS {{\bf S}}
\def\bD {{\bf D}}
\def\bp {{\bf p}}
\def\bP {{\bf P}}
\def\bB {{\bf B}}
\def\bA {{\bf A}}
\def\bC {{\bf C}}
\def\bH {{\bf H}}

\def\cl {{\rm curl}}
\def\dv {{\rm div}}
\def\Cl {{\rm Curl}}
\def\Dv {{\rm Div}}

\def\bi{{\bf i}}

\begin{document}
\title{Mathematical theory for electromagnetic scattering resonances and field enhancement in a subwavelength annular gap}

\author{Junshan Lin$^1$, Wangtao Lu$^2$ and Hai
  Zhang$^3$} 
  \footnotetext[1]{Department of Mathematics, Auburn University,
  Auburn, AL, 36849, USA. Email: jzl0097@auburn.edu. Junshan Lin is partially supported by the NSF grant DMS-2011148.} 
  \footnotetext[2]{School of
  Mathematical Sciences, Zhejiang University, Hangzhou 310027, China. Email:
  wangtaolu@zju.edu.cn. Wangtao Lu is partially supported by NSF of Zhejiang
  Province for Distinguished Young Scholars Grant LR21A010001, NSFC Grant
  12174310, and a Key Project of Joint Funds For Regional Innovation and
  Development (U21A20425).}
\footnotetext[3]{Department of Mathematics, HKUST, Clear Water Bay, Kowloon,
  Hong Kong. Email: haizhang@ust.hk. Hai Zhang is partially supported by Hong Kong RGC grant GRF 16305419 and GRF 16304621.}
\maketitle
\begin{abstract}
This work presents a mathematical theory for electromagnetic scattering resonances in a subwavelength annular hole embedded in a metallic slab, with the annulus width $h\ll1$. The model is representative among many 3D subwavelength hole structures, which are able to induce resonant scattering of electromagnetic wave and the so-called extraordinary optical transmission. We develop a multiscale framework for the underlying scattering problem  based upon a combination of the integral equation in the exterior domain and the waveguide mode expansion inside the tiny hole. The matching of the electromagnetic field over the hole aperture leads to a sequence of decoupled infinite systems, which are used to set up the resonance conditions for the scattering problem. By performing rigorous analysis for the infinite systems and the resonance conditions, we characterize all the resonances in a bounded domain over the complex plane. It is shown that the resonances are associated with the TE and TEM waveguide modes in the annular hole, and they are close to the real axis with the imaginary parts of order ${\cal O}(h)$.  We also investigate the resonant scattering when an incident wave is present. It is proved that the electromagnetic field is amplified with order ${\cal O}(1/h)$ at the resonant frequencies that are associated with the TE modes in the annular hole. On the other hand, one particular resonance associated with the TEM mode can not be excited by a plane wave but can be excited with a near-field electric dipole source, leading to field enhancement of order ${\cal O}(1/h)$.
\end{abstract}

\section{Introduction}
Resonances play a significant role in wave interactions with subwavelength structures, due to their ability to generate unusual physical phenomena that open up  a broad possibilities in modern science and technology. One representative type of resonant subwavelength
structure is nano-holes perforated in noble metals, such as gold or silver. The device of this sort was first introduced in the seminal work \cite{ebblezwol98}, and tremendous research has been sparked since then in pursuit of more efficient resonant nano-hole devices (cf \cite{garmarebbkui10, rodrigo16} and references therein).  The most remarkable phenomenon occurs in these subwavelength devices when an electromagnetic wave is illuminated at the resonant frequencies. The corresponding transmission through the tiny holes exhibit extraordinary large values that can not be explained by the classical diffraction theory developed by Bethe
and is called extraordinary optical transmission (EOT) \cite{ebblezwol98}.
In addition, the EOT is accompanied by strong localized electromagnetic field enhancement inside the subwavelength holes and in the vicinity of hole apertures \cite{garmarebbkui10}. The capability to trigger EOT and to confine light in deep subwavelength apertures leads
to many important applications in biological and chemical sensing, optical lenses, and the design of novel optical devices, etc \cite{blanchard17, cetin2015, huang08, li17plasmonic, rodrigo16}.

The main mechanisms for the EOT and field amplification in the subwavelength hole devices are due to resonances. These include scattering resonances induced by the tiny holes and surface plasmonic resonances generated from the metallic materials \cite{garmarebbkui10}.
Significant progress has been made on the mathematical studies of resonances for two-dimensional subwavelength slit structures in the past few years. In a series of studies, we have established rigorous mathematical theories for a variety of resonances and the induced EOT  via the layer potential technique and asymptotic analysis \cite{linshizha20, linzha17, linzha18a, linzha18b, linzhang21}. The layer potential approach with the operator-based 
Gohberg–Sigal theory was previously used to investigate the resonances in a closely related subwavelength cavity problem \cite{babbontri10, bontri10}. More recently, other mathematical methods were developed to derive the resonances for the two-dimensional slit structures. These include the matched asymptotic method and the Fourier mode matching method \cite{holsch19, zhlu21}. The matched asymptotic expansion techniques have also been applied to construct the solution of the slit scattering problem in \cite{joltor06a, joltor06b, joltor08}. The generalization of the above techniques to the studies of the acoustic wave resonances in three-dimensional subwavelength holes can be found in \cite{fatima21, liazou20, luwanzho21}. We would also like to refer readers to \cite{ammari18, ammari17, ammari16, ammari15} and references therein for the mathematical studies of other type of subwavelength resonances, such as Helmholtz resonators and nanoparticles, etc.

In the previous studies of 2D subwavelength hole resonances or 3D acoustic wave resonances, the governing equations are scalar wave equations. The mathematical studies of electromagnetic resonances for the 3D subwavelength holes remains completely open.
In this paper, we aim to advance the work in this direction by investigating the electromagnetic scattering resonances for the full vector Maxwell's equations. More specifically, we consider the electromagnetic wave scattering by an annular gap, wherein the gap width is much smaller than the incident wavelength. Figure \ref{fig:problem_geometry} depicts the top view and side view of the structure, in which a coaxial waveguide is perforated through a metal slab of thickness $l$, forming an annular gap on the $x_1x_2$ plane.
The annular hole occupies the domain $G^h=R^h\times(-l/2,l/2)$, where
\begin{equation}
  \label{eq:annulus}
  R^h:=\{(x_1,x_2)\in\mathbb{R}^2:x_1=r\cos\theta,x_2=r\sin\theta, r\in(a,a(1+h)),\theta\in[0,2\pi]\}
\end{equation}
denotes the cross-sectional annulus on the $x_1x_2$ plane. In the above, $a$ and $a(1+h)$ are the inner and outer radius of the annulus. In the subsequent analysis, for clarity of presentation we shall scale the geometry of the problem such that $a=1$. The resonances when $a\neq1$ are scaled accordingly by replacing the wavenumber $k$ by $ka$ and the thickness $l$ by $l/a$.  It is assumed that  the gap width is small with $h \ll 1$. The metal region is denoted by 
$$\Omega_{\rm M}:= \big\{(x_1, x_2, x_3) \in\mathbb{R}^3: (x_1,x_2) \in\mathbb{R}^2\backslash \overline{R^h}, x_3\in(-l/2,l/2) \big\}.$$ 
In this work, we focus on the resonances induced by the tiny hole and consider the configuration when the metal is a perfect electric conductor. The studies of plasmonic resonances for real metals and their interactions with the hole resonances are avenues of future research.

\begin{figure}
    \centering
    \includegraphics[width=6.5cm]{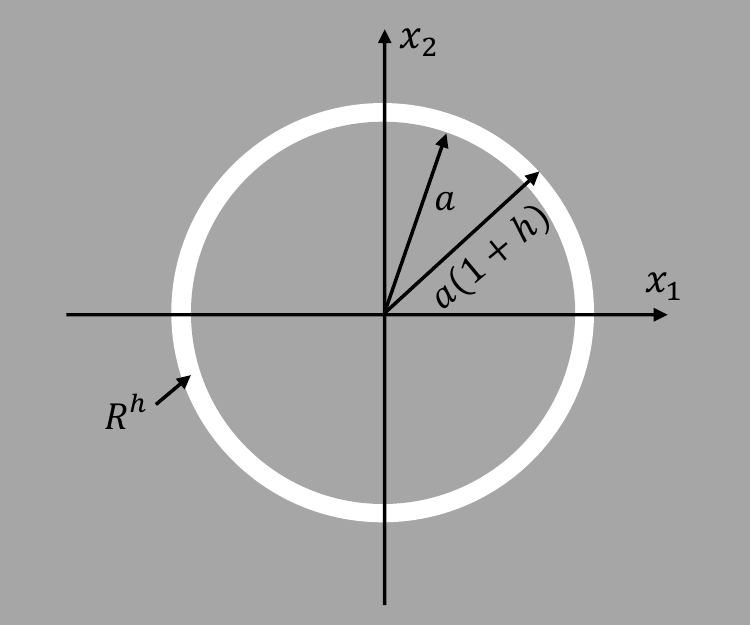} \quad
    \includegraphics[width=6.5cm]{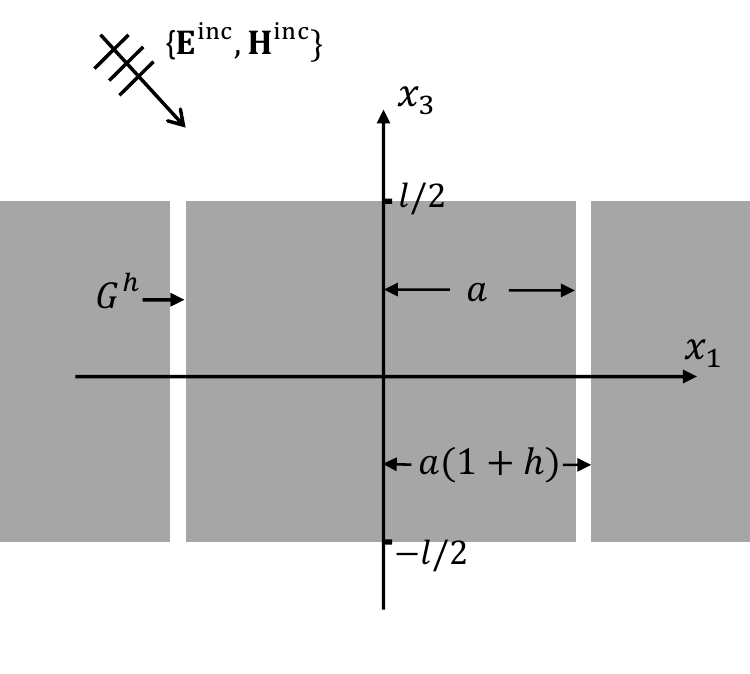}
    \caption{Top view (left) and side view (right) of the subwavelength structure. The cylindrical hole $G^h$ is perforated through the metallic slab with a thickness of $l$ and it forms an annular aperture $R^h$
    on the $x_1x_2$ plane with the inner and outer radius  $a$ and $a(1+h)$ respectively.  }
    \label{fig:problem_geometry}
\end{figure}

Let $\{ {\bf E}^{\rm inc}, {\bf H}^{\rm inc} \}$ be the incoming time-harmonic electromagnetic plane wave given by
\begin{equation}
\label{eq:pla:inc}
{\bf E}^{\rm inc} = {\bf E}^0e^{\bi k x\cdot d},\quad {\bf H}^{\rm inc} = {\bf
H}^0e^{\bi k x\cdot d},
\end{equation}
wherein $k$ is the wavenumber, $d\in\mathbb{R}^3$ is the propagation direction, and the electric and magnetic polarization vectors satisfy ${\bf E}^0\perp d$ and ${\bf H}^0=d\times {\bf E}^0$.
The total electromagnetic field after the scattering
is governed by the following Maxwell's equations:
\begin{align}
  \label{eq:E}
  \cl\ {\bf E} =& \bi k {\bf H}\quad {\rm in}\; \mathbb{R}^3\backslash\overline{\Omega_{\rm M}},\\
  \label{eq:H}
  \cl\ {\bf H} =& -\bi k {\bf E}\quad {\rm in}\; \mathbb{R}^3\backslash\overline{\Omega_{\rm M}},\\
  \label{eq:bc}
  \nu\times{\bf E} =& 0 \quad{\rm on}\; \partial\Omega_{\rm M},
\end{align}
 where $\nu$ denotes the outward unit normal pointing to the exterior of $\Omega_{\rm M}$.
Let $\{{\bf E}^{\rm ref}, {\bf H}^{\rm ref}\}$ be the reflected field above the metal in the absence of the annular hole and $h=0$. The perturbed field generated by the hole $G_h$ when $h>0$ is called scattered field, denoted by $ {\bf E}^{\rm sc} := {\bf E} - {\bf E}^{\rm inc} - {\bf E}^{\rm ref} $ and  $ {\bf H}^{\rm sc} := {\bf H} - {\bf H}^{\rm inc} - {\bf H}^{\rm ref} $. They satisfy the Silver-M\"uller radiation condition (SMC) at infinity above and below the metal (cf \cite{kirhet15}):
\begin{equation}
  \label{eq:smc}
  \lim_{\substack{|x_3|>l/2 \\ |x|\to\infty}}({\bf H}^{\rm sc}\times x - |x| {\bf E}^{\rm sc}) = 0. 
\end{equation}

The resonant phenomena for the scattering problem \eqref{eq:E} - \eqref{eq:smc} was reported and studied experimentally and numerically in \cite{baida02, baida04, hulinluoh18, park15, yoo16}. In this paper, we aim to establish the rigorous mathematical theory for the underlying resonant scattering. The goal is to quantitatively characterize the resonances and study the field enhancement at various resonant frequencies. The mathematical theory presented for this representative structure also seeks to lay the foundational framework in establishing electromagnetic resonant scattering theory for many other 3D subwavelength hole devices to be explored in the future. To this end, we first study the scattering resonances, which lie on the lower complex plane and are the poles the resolvent associated with the scattering problem. The real and imaginary parts of the scattering resonances represent the resonant frequencies and the reciprocal of the resonant magnitude respectively.
The corresponding nontrivial solutions are called quasi-normal modes \cite{dyatlov19}. Equivalently, we consider the homogeneous problem \eqref{eq:E} - \eqref{eq:smc} when the incident field ${\bf E}^{\rm inc} = {\bf H}^{\rm inc} = { \bf 0 }$. The quasi-normal modes satisfy the radiation condition \eqref{eq:smc} but they grow at infinity.  We then study the resonant scattering when the incoming wave attains the resonant frequencies.
The main contribution of this paper is as follows:
\begin{itemize}
    \item[(i)] We prove the existence of electromagnetic scattering resonances for the problem \eqref{eq:E} - \eqref{eq:smc} and present quantitative analysis for the resonances. The structure of resonances is much richer than the resonances for a 2D hole analyzed in \cite{holsch19, linzha17, zhlu21}. In more details, it is shown that the resonances are a sequence of complex numbers that are associated with the TE and TEM waveguide modes in the annular hole. We derive the asymptotic expansion of these resonances. Furthermore, it is demonstrated that the imaginary parts of the resonances attain the order ${\cal O}(h)$. The quantitative analysis of resonances is summarized in Theorems \ref{thm:even:res} and \ref{thm:odd:res}.
    
 \item[(ii)] We also analyze the electromagnetic field governed by \eqref{eq:E}-\eqref{eq:smc} when an incident plane wave is present. We prove that 
electromagnetic field is amplified by order ${\cal O}(1/h)$ at the resonant frequencies that are associated with the TE modes in the annular hole.  A particular resonance associated with the TEM mode can not be excited by the a plane wave. We prove that a near-field electric monopole can be used to excite this resonance to achieve field enhancement of order ${\cal O}(1/h)$. 
The analysis is provided in Section 5 and it explains the observed resonant phenomena through the tiny annular hole reported in \cite{baida02, baida04, hulinluoh18, yoo16}.
\end{itemize}

There are several main challenges in analysis of resonances, due to multiscale nature of the problem and the vector form of the mathematical model. In addition, as elaborated in Section 3, the solution inside the tiny hole consists of several types of waveguide modes (TE, TM and TEM modes), which are responsible for the richness of resonances for the scattering problem. Our multiscale
analysis is based upon a combination of the integral equation formulation with the mode matching method. More precisely, the electromagnetic field outside the annular hole (large-scale domain) is represented by the vector layer potentials and the wave field in the hole (small-scale domain) is expressed as a sum of coaxial waveguide modes, which form a complete basis for the solution space. The matching of the two wave fields for each mode over the annular aperture leads to an infinite system for the expansion coefficients. The main advantage of the mode matching method lies in the natural decoupling of the original system into subsystems with distinct angular momentum in the annulus. Moreover, each individual subsystem can be further reduced into a single nonlinear characteristic equation (resonance condition) by projecting the solution in an infinite-dimensional space onto the dominant resonant modes, and the resonances are the roots of the characteristic equation that can be analyzed by the complex analysis tools. This is achieved by the estimation of the contribution from the modes that are orthogonal to the resonant modes in each subsystem and is accomplished by the asymptotic analysis with respect to the small parameter $h$. The main technical parts are presented in Section 4.

The rest of the paper is organized as follows. In Section 2, we introduce necessary functions spaces and notations to be used throughout the analysis and decompose the whole scattering problem \eqref{eq:E}-\eqref{eq:smc} into two subproblems. The boundary value problems outside and inside the tiny hole are studied in details in Section 3. In particular, we express their solutions via integral equations and the mode expansion respectively. These serve as the starting point for the mode matching framework. Section 4 is devoted to the analysis of scattering resonances. The details of the mode matching formulation, the reduction to the resonance condition in the form of nonlinear characteristic equations, and the analysis of their roots for the complex-valued resonances will be given. Finally, we study the electromagnetic field enhancement at the resonant frequencies in Section 5, and conclude the paper with some discussions in Section 6.

\section{Preliminaries}
\subsection{Function spaces and notations}
We introduce several Sobolev spaces for scalar and vector valued functions that will be used throughout the paper and refer the readers to \cite{mcl00,colkre13,kirhet15} for more details.
Let $\Omega\subset \mathbb{R}^3$ be a bounded Lipschitz domain with the boundary $\Gamma:=\partial\Omega$, and $\nu(x)$ is the unit outward normal on $\Gamma$.
$H^0(\Omega):=L^2(\Omega)$ denotes the set of all
square-integrable functions on $\Omega$. Let $H^1(\Omega)=\{f\in L^2(\Omega):
\nabla f\in [L^2(\Omega)]^3\}$ and $H^{-1}(\Omega)$ be its dual space.
$H^s(\Omega)$ denotes the fractional Sobolev space for $-1<s<1$,. 
Given $\Gamma_1\subset \Gamma$, we define $H^{s}(\Gamma_1)$ by $H^{s}(\Gamma_1)=\{f|_{\Gamma_1}:f\in H^s(\Gamma)\}$
and its dual space by
$$ [H^{s}(\Gamma_1)]'=\widetilde{H^{-s}}(\Gamma_1):=\{f\in
H^{-s}(\Gamma): {\rm supp} f\subset \overline{\Gamma_1}\}.$$
Here $H^s(\Gamma)$ is the Sobolev space over the boundary $\Gamma$.

For a vector valued function ${\bf F}(x)=[F_1(x),F_2(x),F_3(x)]^{T}$ with components $F_j\in\mathbb{C}, j=1,2,3$, $\cl\ \bF=\nabla \times \bF$ and
$\dv\ \bF=\nabla\cdot \bF$ denote the curl and the divergence of ${\bf F}$,
respectively. Let 
$$H(\cl,\Omega):=\{{\bf F}\in [L^2(\Omega)]^3: \cl\ {\bf F}\in [L^2(\Omega)]^3\}.$$
We also define
$$ H_t^{s}(\Gamma)=\{\bF\in[H^{s}(\Gamma)]^3: \nu\cdot \bF=0\} \quad \mbox{for} -1/2\leq s\leq 1/2, $$  and $L_t^2(\Gamma)=H_t^{0}(\Gamma)$. Let $\Cl\ F$ and $\Dv\ F$ be
the surface divergence and the surface curl of $F$ on $\Gamma$, respectively
(c.f. Eqs.~(6.37) and (6.41) in \cite{colkre13}). 
For the planar surfaces $R^h \times \{x_3=\pm l/2\}$ considered in this paper, we have
$$
 \Dv = \nabla_2\cdot=[\partial_{x_1},\partial_{x_2}]^{T}\cdot \; , \quad \Cl =\cl_2=
[\partial_{x_2},-\partial_{x_1}]^{T}\cdot \;.
$$

Define $$H^{-1/2}(\Dv,
\Gamma)=\{\bF\in H_t^{-1/2}(\Gamma): \Dv\ \bF\in H^{-1/2}(\Gamma) \}$$ and
$$ H^{-1/2}(\Cl,\Gamma)=\{\bF\in H_t^{-1/2}(\Gamma):\Cl\ \bF\in H^{-1/2}(\Gamma)\}. $$
By \cite[Thm. 5.26]{kirhet15}, $H^{-1/2}(\Cl,\Gamma) = [H^{-1/2}(\Dv,
\Gamma)]'$ where the duality is defined by
\begin{equation}
  \label{eq:bl}
  \bF(\bG) = \int_{\Gamma} \bF\cdot\bG \, ds(\Gamma)
\end{equation}
for any $\bF\in H^{-1/2}(\Cl,\Gamma)$ and $\bG\in H^{-1/2}(\Dv,\Gamma)$.
From the trace theorem \cite[Thm. 5.24]{kirhet15}, the trace operators 
$$\gamma_t:
H(\cl,\Omega)\to H^{-1/2}(\Dv,\Gamma), \bF\mapsto \nu\times \bF|_{\Gamma}$$
and
$$\gamma_T: H(\cl,\Omega)\to H^{-1/2}(\Cl,\Gamma), \bF\mapsto
(\nu\times\bF|_{\Gamma})\times \nu$$ are bounded. 
Given an open domain $\Gamma_1\subset \Gamma$, we define  
$$ H^{-1/2}(\Cl,\Gamma_1)=\{\bF|_{\Gamma_1}: \bF\in
H^{-1/2}(\Cl,\Gamma)\}$$ 
and its dual space
$$\tilde{H}^{-1/2}(\Dv,\Gamma_1)=\{\bF\in H^{-1/2}(\Dv,\Gamma): {\rm
  supp}\,\bF\subset\overline{\Gamma_1}\}. 
$$
Finally, for an unbounded Lipschitz
domain $\Omega$, we let
$$H_{\rm loc}(\cl,\Omega):=\{\bF: \bF|_{\Omega\cap
  B(0,r)}\in H(\cl,\Omega\cap B(0,r))\textrm{\ for any\ }r>0\}$$ 
wherein $B(0,r):=\{x:|x|<r\}$.

We also introduce the following notations for the problem geometry to be used in the rest of the paper:
\begin{itemize}
    \item[(1)]$\Omega_\pm=\{x\in\mathbb{R}^3\backslash\overline{\Omega_M}:\pm x_3>0\}$: the upper and lower half domain exterior to the metal;
    \item[(2)] $\mathbb{R}^{3}_+=\{x \in \mathbb{R}^3 : x_3>l/2 \}$: the half space above the metal;
    \item[(3)] $G^{h}_+ = \{ x \in G_h: x_3>0 \}$: the upper half of the annular hole $G_h$;
    \item[(4)] $ A^h = \{x: (x_1,x_2)\in R^h, x_3=l/2\}$: the upper annular aperture of $G^{h}_+$;
    \item[(5)] $\Gamma_b=\{ x: (x_1,x_2)\in R^h, x_3=0\}$: the annulus on the $x_1x_2$ plane or the base of $G^{h}_+$;
    \item[(6)] $\Gamma^{h}_{+}$: the side boundary of $G^{h}_+$.
\end{itemize}
In addition, the following sets will be used:
\begin{itemize}
\item[(1)]
 ${\cal B} := \{z\in \mathbb{C}:
|z|<C_0\}$, where $C_0$ is a fixed positive constant;
\item[(2)] $\mathbb{N}^* := \{1, 2, 3, \cdots.\}$;
\item[(3)] $(\mathbb{Z}\times\mathbb{N})^*:=(\mathbb{Z}\times\mathbb{N})\backslash\{(0,0)\}$.
\end{itemize}
Finally, $ A \eqsim B$ implies $c_1 B \leq A \leq c_2B$ for some positive constants $c_1, c_2 $ that are independent of $A$ and $B$.

\subsection{Decomposition of the scattering problem}
Due to the symmetry of the structure with respect to the $x_1x_2$ plane, the  scattering problem (\ref{eq:E})-(\ref{eq:smc}) can be decomposed as the two subproblems as
follows:
\begin{itemize}
\item[(E).] Given the incident field $[\bE^{\rm inc},\bH^{\rm inc}]/2$, solve for
    $[\bE^{\rm e},\bH^{\rm e}]$ that satisfies  
    \begin{align}
      \cl\ \bE^{\rm e} =& \, \bi k \bH^{\rm e}\quad{\rm in}\; \Omega_+, \label{eq:problemE_1} \\
      \cl\ \bH^{\rm e} =& \, -\bi k \bE^{\rm e}\quad{\rm in}\; \Omega_+, \label{eq:problemE_2}\\
      \nu\times {\bE}^{\rm e} =& \, 0\quad{\rm on}\; \partial\Omega_+\backslash\Gamma_b, \label{eq:problemE_3} \\
      \nu\times {\bH}^{\rm e} =& \, 0 \quad{\rm on}\; \Gamma_b, \label{eq:problemE_4}
    \end{align}
    and the radiation condition (\ref{eq:smc}) for $x_3\ge l/2$ with $[\bE^{\rm
      sc},\bH^{\rm sc}]=[\bE^{\rm e},\bH^{\rm e}] - [\bE^{\rm inc},\bH^{\rm inc}]/2 - [\bE^{\rm ref},\bH^{\rm
      ref}]/2$.
  \item[(O).] Given the incident field $[\bE^{\rm inc},\bH^{\rm inc}]/2$,  solve for $[\bE^{\rm o},\bH^{\rm o}]$ that satisfies  
    \begin{align}
      \cl\ \bE^{\rm o} =& \, \bi k \bH^{\rm o}\quad{\rm in}\; \Omega_+, \label{eq:problemO_1} \\
      \cl\ \bH^{\rm o} =& \, -\bi k \bE^{\rm o}\quad{\rm in}\; \Omega_+,  \label{eq:problemO_2} \\
      \nu\times {\bE}^{\rm o} =& \, 0\quad{\rm on}\; \partial\Omega_+, \label{eq:problemO_3}
    \end{align}
    and the radiation condition (\ref{eq:smc}) for $x_3\ge l/2$ with $[\bE^{\rm sc},\bH^{\rm
      sc}]=[\bE^{\rm o},\bH^{\rm o}] - [\bE^{\rm inc},\bH^{\rm inc}]/2 - [\bE^{\rm ref},\bH^{\rm ref}]/2$.
\end{itemize}
It is clear that the solution of the scattering problem \eqref{eq:E} - \eqref{eq:smc} can be written as
\begin{align}
  \bE(x) = \left\{
  \begin{array}{lc}
    \bE^{\rm e}(x) + \bE^{\rm o}(x), & x_3\geq 0,\\
     \bE^{\rm e,*}(x^*) - \bE^{\rm o,*}(x^*), & x_3<0,\\
  \end{array}
  \right.
  \quad \bH(x) = (\bi k)^{-1}\cl\ \bE.
\end{align}
In the above, $*$ denotes the reflection vector with respect to the $x_1x_2$ plane. On the other hand, there holds
\begin{align}
{\bE}^{\rm e}(x) =& \frac{\bE(x)+\bE^*(x^*)}{2},\quad   \bE^{\rm o}(x) = \frac{\bE(x)-\bE^*(x^*)}{2},\quad x_3>0, \\
  \bH^{\rm j}(x) =& (\bi k)^{-1}\cl\ \bE^{\rm j},\quad {\rm j}\in\{{\rm e},{\rm o} \}.
\end{align}

In the rest of the paper, for clarity we shall present the detailed analysis for the resonances for Problem (E) only. Problem (O) can be analyzed similarly, thus we will point out the main difference in the analysis and present the main results directly. To simplify the notations, we shall overload $\bE$ and $\bH$ for $\bE^{\rm e}$ and $\bH^{\rm e}$, respectively.


\section{Two auxiliary boundary value problems}
In this section, we study the exterior boundary value problem above the metal
and the interior boundary value problem in the annular hole.
They will serve as the foundation for the mode matching framework and for establishing the resonance condition for the scattering problem (E). The notations introduced in Section 2.1 for the problem geometry are used.


\subsection{Scattering problem above the metal}
For a given vector valued function $\bF$ on $A^h$, let
\begin{align}
  \tilde{\cal L}_k[\bF](x) =& \cl\ \cl\int_{A^h}\Phi_k(x;y)\bF(y)ds(y),\\
  \tilde{\cal M}_k[\bF](x) =& \cl \int_{A^h}\Phi_k(x;y)\bF(y)ds(y),
\end{align}
be the vector layer potentials for $x\in\mathbb{R}_{+}^3$,  where
$\Phi_k(x;y)=\frac{e^{\bi k|x-y|}}{4\pi|x-y|}$ for $x\neq y$.
Consider the following half-space problem above the metal
\begin{align*}
  ({\rm HSP}):\quad\quad \left\{
  \begin{array}{ll}
  \cl\ \bE = \bi k \bH,\quad&{\rm in}\quad \mathbb{R}_{+}^3,\\
  \cl\ \bH = -\bi k \bE,\quad&{\rm in}\quad \mathbb{R}_{+}^3,\\
  \nu\times\bE=0,\quad &{\rm on}\quad \{x\in\mathbb{R}^3:x_3=l/2\}\backslash \overline{A^{h}},\\
  \nu\times\bE=\bF,\quad &{\rm on}\quad A^{h},\\
  \end{array}
  \right.
\end{align*}
with the radiation condition (\ref{eq:smc}) in $x_3>l/2$. The following theorem
states the well-posedness of the problem.
\begin{mytheorem}
  \label{thm:wp:hsp}
  For any $k>0$ and any $\bF\in \tilde{H}^{-1/2}(\Dv, A^{h})$, the
  following two functions
    \begin{align}
      \label{eq:EH:hsp}
      \bE = -2 \tilde{M}_k [\bF],\quad \bH = -\frac{2}{\bi k}\tilde{L}_k[\bF],
    \end{align}
    in $H_{\rm loc}(\cl,\mathbb{R}_{+}^3)$ constitute the unique solution to
    problem (HSP).
  \begin{proof}
    For $\bF\equiv 0$, one follows Lemma 5.30 in \cite{kirhet15} to extend
    $[\bE,\bH]$ to be a function in $[H_{\rm loc}(\cl,\mathbb{R}^3)]^2$, 
    which satisfies the radiation condition (\ref{eq:smc}) in all directions $x/|x|$. Thus, $\bE\equiv
    \bH\equiv 0$ so that (HSP) has at most one solution for $\bF\neq 0$. One
    directly verifies that $[\bE,\bH]$ in (\ref{eq:EH:hsp}) is the unique
    solution of (HSP) in $[H_{\rm loc}(\cl,\mathbb{R}_{+}^3)]^2$.
  \end{proof}
\end{mytheorem}
Let ${\cal L}_k[{\bF}]$ be the trace of $\tilde{\cal L}_k[\bF]$ on $A^h$. By
Theorem~\ref{thm:wp:hsp} and the open mapping theorem, 
\begin{align}
  \label{eq:t2t:0}
  \nu\times \bH|_{A^h} = \frac{-2}{\bi k}{\cal L}_k[\nu\times \bE|_{A^h}]\in H^{-1/2}(\Dv, A^h),
\end{align}
and ${\cal L}_k$ is bounded from $\tilde{H}^{-1/2}(\Dv,A^h)$ to $H^{-1/2}(\Dv,A^h)$. Clearly,
${\cal L}_k$ maps the tangential component of $\bE$ to that of $\bH$ so that we
shall call it the tangential-to-tangential (T2T) map in the sequel. As we
shall see in Section 4, the T2T map ${\cal L}_k$ plays a central
role in formulating the resonance eigenvalue problem.

\subsection{Boundary value problem in the annular hole}

Recall that $A_h$ denotes the planar annular aperture. In this section, we first construct
a countable basis for the function space $\widetilde{H}^{-1/2}(\Dv,A^h)$ and then express the solution of the boundary value problem in the annular hole $G^{h}_+$ using the basis.

\subsubsection{A complete basis for $\widetilde{H}^{-1/2}(\Dv,A^h)$}
As $[L^2(A^h)]^2\cap \widetilde{H}^{-1/2}(\Dv,A^h)$ is dense in
$\widetilde{H}^{-1/2}(\Dv,A^h)$, we only need to construct a dense and countable
basis of $[L^2(A^h)]^2$. From Eqs.~(1.42)\&(1.55) in Chpt. IX of
\cite{daulio90}, the Helmholtz decomposition of $[L^2(A^h)]^2$ is given below.
\begin{mylemma}
  \label{lem:helmdecomp}
  Let $\Delta_2=\nabla_2\cdot\nabla_2$, and 
  \begin{align}
\cl_2\ H^1(A^h):=&\{\cl_2 f: f\in H^1(A^h)\},\\
\nabla_2 H_0^1(A^h):=&\{\nabla_2f: f\in H_0^1(A^h)\},\\
    \mathbb{H}_2(A^h):=&\{\nabla_2 f: f\in H^1(A^h), \Delta_2 f = 0, f|_{r=a}=C_1, f|_{r=a(1+h)}=C_2; C_1, C_2\in\mathbb{R}\},
\end{align}
be three closed subspaces of $[L^2(A^h)]^2$ that are orthogonal to each
other in the sense of the $L^2$-inner product.
Then, $[L^2(A^h)]^2$ can be decomposed into the direct sum
of the above three subspaces, i.e.,
\begin{equation}
  [L^2(A^h)]^2 = \cl_2\ H^1(A^h) \oplus \nabla_2 H_0^1(A^h)\oplus\mathbb{H}_2(A^h).
\end{equation}
\end{mylemma}
Now we find a countable basis for each of the three subspaces. It is
not hard to see that one-dimensional $\mathbb{H}_2(A^h)$ is given by
\begin{equation}
  \label{eq:bs:h2}
  \mathbb{H}_2(A^h) = {\rm span}\big\{\nabla_2\log(r)\big\}. 
\end{equation}
To characterize $\nabla_2 H_0^1(A^h)$, we consider the following Dirichlet eigenvalue problem
\begin{align*}
  {\rm (DEP):}\quad\quad \left\{
  \begin{array}{ll}
  -\Delta_2\psi = \lambda \psi \quad&{\rm in}\; R^h,\\
  \psi = 0 \quad&{\rm on}\; \partial R^h.\\
  \end{array}
\right.
\end{align*}
The countable normalized eigenfunctions are (cf. \cite{kutsig84}) 
\begin{align}
  \label{eq:phimnD}
  \psi_{ij}^{D}(r,\theta;h) :=& \left[ C_{ij}^{D} \right]^{-1}\Big[ Y_{|i|}(\beta_{|i|j}^{D})J_{|i|}(\beta_{|i|j}^{D}r)-J_{|i|}(\beta_{|i|j}^{D})Y_{|i|}(\beta_{|i| j}^{D}r) \Big]e^{\bi i\theta},
\end{align}
for $(i,j)\in\mathbb{Z}\times \mathbb{N}^*$. The associated eigenvalues 
\begin{equation} \label{eq-lambdaD}
  \lambda_{ij}^{D} = \left( \beta_{|i| j}^{D}\right)^2>0.
\end{equation}
In the above, $J_i$ and $Y_i$ are the first and second kind Bessel functions of order $i$, $\beta_{|i|j}^{D}$ is the $j$-th positive root (in ascending order) of the following euqation:
\begin{equation}
  \label{eq:gov:roots}
  F^{D}_{|i|}(\beta;h):=Y_{|i|}(\beta)J_{|i|}\left( \beta(1+h) \right) - J_{|i|}(\beta)Y_{|i|}\left( \beta(1+h) \right)=0,
\end{equation}
and $C_{ij}^{D}>0$ is chosen such that $||\psi_{ij}^{\rm
  D}||_{L^2(R^h)}=1$. 

For $0<h\ll1$, the asymptotic analysis of $\lambda_{ij}^{D}$ and $\psi_{ij}^{D}$ are carried out in detail in Appendix A. It is shown in \eqref{eq:asy:betamn} and \eqref{eq:psimnD} that 
$$\lambda_{ij}^{D}\sim \left(\frac{j\pi}{h}\right)^2 \quad \mbox{and} \quad \psi_{ij}^D \sim \frac{e^{im\theta}}{\sqrt{\pi r h}} \sin\left(\frac{n\pi}{h}(r-1)\right)$$ for $h\ll1$.
It follows from \cite[Thm. 4.12]{mcl00} that
$\{\psi_{ij}^{D}(\cdot;h)\}_{(i,j)\in\mathbb{Z}\times\mathbb{N}^*}$  constitutes a complete
orthonormal basis of $L^2(R^h)$ vanishing on the boundary $\partial R^h$, and they form a dense and countable basis for $H_0^1(R^h)$. Therefore,
\begin{equation}
  \label{eq:bs:gr}
  \nabla_2 H_0^1(R^h)=\overline{{\rm span}\{\nabla\psi_{ij}^{D}(\cdot;h)\}_{(i,j)\in\mathbb{Z}\times\mathbb{N}^*}}, 
\end{equation}
where the overline denotes the closure.

As for the subspace $\cl_2\ H^1(A^h)$, we consider the following Neumann eigenvalue problem
\begin{align*}
  {\rm (NEP):}\quad\quad \left\{
  \begin{array}{ll}
  -\Delta_2\psi = \lambda \psi,\quad&{\rm in}\quad R^h\\
  \partial_{\nu}\psi = 0\quad&{\rm on}\quad \partial R^h.\\
  \end{array}
\right.
\end{align*}
As shown in \cite{kutsig84}, the countable
eigenvalues are $\lambda^{N}_{00}:=0$ and 
\begin{equation}\label{eq-lambdaN}
  \lambda_{mn}^{N} = \left( \beta_{|m|n}^{N}\right)^2,
  \quad (m,n)\in(\mathbb{Z}\times\mathbb{N})^*,
\end{equation}
where $\beta_{|m|n}^{N}$ is the $n$-th nonnegative root (in ascending order starting from $n=0$) of the equation
\begin{equation}
  \label{eq:gov:roots:n}
  F_{|m|}^{N}(\beta;h):=Y_{|m|}'(\beta)J_{|m|}'\left( \beta(1+h) \right) - J_{|m|}'(\beta)Y_{|m|}'\left( \beta(1+h) \right)=0.
\end{equation}
The associated normalized eigenfunctions are
$\psi^{N}_{00}:=\frac{1}{\sqrt{\pi h(2+h)}}$ and
\begin{align}\label{eq:psimnN}
  \psi_{mn}^{N}(r,\theta;h) :=& \left[ C_{mn}^{N} \right]^{-1}\Big[ Y_{|m|}'(\beta_{|m|n}^{N})J_{|m|}(\beta_{|m|n}^{N}r)-J_{|m|}'(\beta_{|m|n}^{N})Y_{|m|}(\beta_{|m|n}^{N}r) \Big]e^{\bi m\theta}, 
\end{align}
in which $C_{mn}^{N}>0$ is chosen such that $||\psi_{mn}^{\rm N}||_{L^2(R^h)}=1$.

The asymptotic formulas of $\lambda_{mn}^{N}$ and
  $\psi_{mn}^{N}$ are provided in \eqref{eq:asy:betamn:N} - \eqref{eq:psimnN:n0}. It is important to note that
   when $h\ll 1$,
  $$\lambda_{m0}^{N} \sim m^2 \quad \mbox{while} \quad  \lambda_{mn}^{N} \sim (\frac{n\pi}{h})^2 \quad \mbox{for} \; n\ge1.
  $$
 The eigenfunctions
  $$\psi_{m0}^{\rm N}\sim\frac{e^{\bi m\theta}}{\sqrt{\pi h(h+2)}} \quad\mbox{and}\quad \psi_{mn}^{N} \sim \frac{e^{\bi m\theta}}{\sqrt{\pi rh}}\cos\left[ \frac{n\pi}{h}(r-1) \right] \quad \mbox{for} \; n\ge1. $$
$\{\psi_{mn}^{\rm N}\}_{m\in\mathbb{Z},n\geq 0}$ constitutes a
complete orthonormal basis of $L^2(R^h)$ (cf. \cite[Thm. 4.12]{mcl00}), and is a dense and countable
basis of $H^1(A^h)$. Therefore,
\begin{equation}
  \label{eq:bs:Cl}
  \cl_2\ H^1(A^h) = \overline{{\rm span}\{\cl_2\ \psi_{mn}^{N}(\cdot;h)\}_{(m,n)\in(\mathbb{Z}\times\mathbb{N})^*}}.
\end{equation}
where we have excluded the constant eigenfunction $\psi_{00}^{N}$.

To ease the burden of notations in the subsequent analysis, we introduce the rotation operator ${\cal R}$:
\begin{equation}\label{eq:R}
  {\cal R}: f=[f_1,f_2]\mapsto [f_2,-f_1],\quad\forall f\in [L^2(A^h)]^2.
\end{equation}
It is clear that ${\cal R}[L^2(A^h)]^2=[L^2(A^h)]^2$.
Consequently, by virtue of (\ref{eq:bs:h2}), (\ref{eq:bs:gr}) and
(\ref{eq:bs:Cl}), we have
\begin{align}
  \label{eq:bs:tot}
  \widetilde{H}^{-1/2}(\Dv,A^h) =& \overline{{\cal R} [L^2(A^h)]^2\cap\widetilde{H}^{-1/2}(\Dv,A^h)}\nonumber\\
  =& \overline{{\rm span}\{{\cal R}\cl_2\ \psi_{mn}^{N}, {\cal R}\nabla \psi_{ij}^{D}, {\cal R}\nabla\log r\}_{(m,n,i,j)\in(\mathbb{Z}\times\mathbb{N})^*\times \mathbb{Z}\times\mathbb{N}^*}}, 
\end{align}
where the norm of $\widetilde{H}^{-1/2}(\Dv,A^h)$ is used for the completion. In the next
subsection, we construct the solution in $G^{h}_+$ for $\nu\times
\bE|_{A^{h}}$ being
one of the basis functions.

\subsubsection{Field representation in the annular hole}
Given $\bF\in \widetilde{H^{-1/2}}(\Dv,A^h)$, let us consider the boundary value problem
\begin{align*}
({\rm AHP}):\quad\quad \left\{
  \begin{array}{l}
  \cl\ \bE = \bi k \bH \quad{\rm in} \; G^{h}_+,\\
  \cl\ \bH = -\bi k \bE \quad{\rm in} \; G^{h}_+,\\
  \nu\times \bE|_{\Gamma^{h}_{+}} = 0,\\
  \nu\times \bH|_{\Gamma_b} = 0,\\
  \nu\times \bE|_{A^h} = \bF. \\
  \end{array}
  \right.
\end{align*}
The well-posedness of problem (AHP) is given in the following theorem.
\begin{mytheorem}
  \label{thm:wp:hlp}
    Assume that $k$ is real and positive and $k\notin\{\sqrt{\lambda_{mn}^N+(2i+1)^2\pi^2/l^2}:
  (m,i,n)\in\mathbb{Z}^2\times\mathbb{N}\}$. 
  For $0<h\ll 1$, the boundary value
  problem (AHP) attains a unique solution $[\bE,\bH]\in H(\cl, G^{h}_+)$ that
  depends continuously on the boundary data $\bF\in \widetilde{H^{-1/2}}(\Dv,A^h)$.
  \begin{proof}
    We first address the uniqueness. Let $\bF\equiv 0$. It follows from Lemma 5.30(b) of \cite{kirhet15} that an even reflection of
    $\bE$ extends $\bE$ and
    $\bH$ into $G^h$ such that
    \begin{align*}
  \cl\ \bE =& \bi k \bH\;{\rm in}\quad G^{h},\\
  \cl\ \bH =& -\bi k \bE\;{\rm in}\quad G^{h},\\
  \nu\times \bE =& 0\quad{\rm on}\; \partial G^{h}.
    \end{align*}
    An odd reflection of $\bE$ w.r.t $x_3=\pm l/2$ extends both $\bE,\bH$ to a
    larger domain $\Omega$ with $G^h\subset \Omega$ and that
    $[\bE,\bH]\in [H(\cl,\Omega)]^2$ with $\nu\times \bE=0$ on $\partial\Omega$.
    It can be shown that $\bE,\bH\in [H^1(G^h)]^3$; see, for instance, \cite[Chapter IX,\S 1]{daulio90}. 
   Thus $E_3\in H^1(G^h)$ satisfies
    \begin{align*}
      -\Delta E_3 =& k^2 E_3 \quad{\rm in}\; G^h,\\
      E_3 =& 0 \quad{\rm on} \; \Gamma^h,\\
      \partial_{\nu} E_3 =& 0 \quad {\rm on} \; A^{h}\cup A^{h}_{-},
    \end{align*}
    where $A^{h}_{-}:=\{x:(x_1,x_2)\in R^h,x_3 = -l/2\}$ is the bottom aperture of $G^h$ and $\Gamma^h$ is the side boundary  of $G^h$. 
    In light of \eqref{eq:asy:betamn}, we choose sufficiently small $h$ such that $k^2$ is not an eigenvalue of
    the above problem. Consequently, $E_3\equiv 0$ in $G^h$.

    Next, $H_3\in H^1(G^h)$ satisfies
    \begin{align*}
      -\Delta H_3 =& k^2 H_3 \quad{\rm in}\; G^h,\\
      \partial_{\nu}H_3 =& 0 \quad{\rm on}\; \Gamma^h,\\
       H_3 =& 0 \quad {\rm on}\; A^{h}\cup A^{h}_{-}.
    \end{align*}
    The boundary value problem attains trivial solution when $k\notin\{\sqrt{\lambda_{mn}^N+(2i+1)^2\pi^2/l^2}:
  (m,i,n)\in\mathbb{Z}^2\times\mathbb{N}\}$.
  Therefore,  $E_1$ and $E_2$ can be expressed as in the form of 
    \[
      \left[
        \begin{array}{c}
          E_1(x)\\
          E_2(x)\\
          \end{array}
      \right] =       \left[
        \begin{array}{c}
          f_1(x_1,x_2)\\
          f_2(x_1,x_2)\\
          \end{array}
      \right]e^{\bi kx_3} + 
\left[
        \begin{array}{c}
          g_1(x_1,x_2)\\
          g_2(x_1,x_2)\\
          \end{array}
      \right]e^{-\bi kx_3},
    \]
where $f_j$ and $g_j$ ($j=1,2$) are harmonic functions, and 
    \[
      \Dv \left[
        \begin{array}{c}
          f_1(x_1,x_2)\\
          f_2(x_1,x_2)\\
          \end{array}
      \right] = \Dv \left[
        \begin{array}{c}
          g_1(x_1,x_2)\\
          g_2(x_1,x_2)\\
          \end{array}
      \right] =       \Cl \left[
        \begin{array}{c}
          f_1(x_1,x_2)\\
          f_2(x_1,x_2)\\
          \end{array}
      \right] = \Cl \left[
        \begin{array}{c}
          g_1(x_1,x_2)\\
          g_2(x_1,x_2)\\
          \end{array}
      \right] =  0.
    \]
    By Lemma~\ref{lem:helmdecomp}, it can be verified that
    $[f_1,f_2]^{T},[g_1,g_2]^{T}\in\mathbb{H}_2$. Consequently, 
    \begin{align*}
  E_1 =& \frac{x_1}{x_1^2+x_2^2}(c_1e^{\bi k x_3} + c_2e^{-\bi k x_3}),\quad E_2 = \frac{x_2}{x_1^2+x_2^2}(c_1e^{\bi k x_3} + c_2e^{-\bi k x_3}),\\
 H_1 =& \frac{-x_2}{x_1^2+x_2^2}(c_1e^{\bi k x_3} - c_2e^{-\bi k x_3}),\quad H_2 = \frac{x_1}{x_1^2+x_2^2}(c_1e^{\bi k x_3} -c_2e^{-\bi k x_3}),
\end{align*}
for some constants $c_1$ and $c_2$. The boundary condition $E_1=E_2=0$ on
$A^h\cup A^{h}_{-}$ implies
\[
  c_1 e^{\bi kl/2} + c_2e^{-\bi kl/2} = 0,\quad c_1=c_2.
\]
Thus a nonzero solution $[\bE,\bH]$ exists if and only if $e^{\bi kl/2}+e^{-\bi
  kl/2}=2\cos(kl/2)=0$, which is excluded by our assumption. Now the
well-posedness follows thanks to Theorem 5.60 in \cite{kirhet15}.
  \end{proof}
\end{mytheorem}

We now construct special solutions to the problem (AHP), which are called waveguide modes in the annular hole $G^h_{+}$. Denote
\begin{equation} \label{eq-s}
s_{mn}^N=\sqrt{k^2-\lambda_{mn}^N}, \quad s_{ij}^D = \sqrt{k^2-\lambda_{ij}^D}.
\end{equation} 
Assume that $k\notin\{\sqrt{\lambda_{mn}^N+(2i+1)^2\pi^2/l^2}$. 
\begin{itemize}
  \item[1.] Transverse electric (TE) modes.  For each $(m,n)\in(\mathbb{Z}\times\mathbb{N})^*$, define 
    \begin{align}\label{eq:E_TE_modes}
      \bE_{mn}^{TE} =& \left[
      \begin{array}{l}
        (e^{\bi s_{mn}^N x_3} + e^{-\bi s_{mn}^N x_3})\partial_{x_2}\psi_{mn}^N\\
        -(e^{\bi s_{mn}^N x_3} + e^{-\bi s_{mn}^N x_3})\partial_{x_1}\psi_{mn}^N\\
        0
      \end{array}
      \right],
    \end{align}
    \begin{align}\label{eq:H_TE_modes}
      \bH_{mn}^{TE} =& \frac{1}{k}\left[
      \begin{array}{l}
        s_{mn}^N(e^{\bi s_{mn}^N x_3} - e^{-\bi s_{mn}^N x_3})\partial_{x_1}\psi_{mn}^N\\
        s_{mn}^N(e^{\bi s_{mn}^N x_3} - e^{-\bi s_{mn}^N x_3})\partial_{x_2}\psi_{mn}^N\\
        -\bi \lambda_{mn}^{N}(e^{\bi s_{mn}^N x_3} + e^{-\bi s_{mn}^N x_3})\psi_{mn}^{N}. 
      \end{array}
      \right]
    \end{align}
Then $\{[\bE_{mn}^{TE},\bH_{mn}^{TE}]\}_{(m,n)\in(\mathbb{Z}\times\mathbb{N})^*}$ is the unique solution of (AHP) with $\bF=[F_1,F_2, 0]^T =[ 2\cos(s_{mn}^Nl/2){\cal R}\cl_2\ \psi_{mn}^N, 0]^T$. These solutions
are called transverse electric (TE) modes.
  \item[2.] Transverse magnetic (TM) modes. For each $(i,j)\in\mathbb{Z}\times\mathbb{N}^*$, define
    \begin{align}\label{eq:E_TM_modes}
      \bE_{ij}^{TM} =& \left[
      \begin{array}{l}
        (e^{\bi s_{ij}^D x_3} + e^{-\bi s_{ij}^D x_3})\partial_{x_1}\psi_{ij}^D\\
        (e^{\bi s_{ij}^D x_3} + e^{-\bi s_{ij}^D x_3})\partial_{x_2}\psi_{ij}^D\\
        \lambda_{ij}^{D}/(\bi s_{ij}^D)(e^{\bi s_{ij}^D x_3} -e^{-\bi s_{ij}^D x_3})\psi_{ij}^{D}\\
      \end{array}
      \right],
      \end{align}
      \begin{align}\label{eq:H_TM_modes}
      \bH_{ij}^{TM} =& \frac{k(e^{\bi s_{ij}^D x_3} - e^{-\bi s_{ij}^D x_3})}{s_{ij}^D}\left[
      \begin{array}{l}
        -\partial_{x_2}\psi_{ij}^D\\
        \partial_{x_1}\psi_{ij}^D\\
        0\\
      \end{array}
      \right].
    \end{align}
Then $\{[\bE_{ij}^{TM},\bH_{ij}^{TM}]\}_{(i,j)\in\mathbb{Z}\times\mathbb{N}^*}$  is the unique solution of (AHP) with $\bF= [2\cos(s_{ij}^Dl/2){\cal R}\nabla_2\psi_{ij}^D,0]^T$. These solutions are
called transverse magnetic (TM) modes.
  \item[3.] Transverse electromagnetic
    (TEM) mode.  Define
    \begin{align}\label{eq:E_TEM_modes}
      \bE_{E}^{TEM} =& (e^{\bi k x_3} + e^{-\bi k x_3})\left[
      \begin{array}{l}
        \partial_{x_1}\log r\\
        \partial_{x_2}\log r\\
        0\\
      \end{array}
      \right],
      \end{align}
      \begin{align}\label{eq:H_TEM_modes}
      \bH_{E}^{TEM} =& (e^{\bi k x_3} - e^{-\bi k x_3})\left[
      \begin{array}{l}
        -\partial_{x_2}\log r\\
        \partial_{x_1}\log r\\
        0\\
      \end{array}
      \right].
    \end{align}
Then $\{[\bE_{E}^{TEM},\bH_{E}^{TEM}]\}$ is the unique solution of (AHP) with $\bF= [2\cos(kl/2){\cal R}\nabla_2\log r,0]^T$. This solution is called the transverse electromagnetic (TEM) mode.
\end{itemize}
\begin{myremark}
  For $k\in\{\sqrt{\lambda_{mn}^N+(2i+1)^2\pi^2/l^2}:
  (m,i,n)\in\mathbb{Z}^2\times\mathbb{N}\}$, we use
  $\nu\times\bH|_{A^h}=\bF^H=[F_1^H,F_2^H,0]$ as the boundary condition instead,
  where we choose $[F_1^H,F_2^H]$ from $\Big\{\frac{2\bi
    s_{mn}^N}{k}\sin(s_{mn}^Nl/2)\nabla_2\psi_{mn}^N, \frac{-2k\bi
    \sin(s_{ij}^Dl/2)}{s_{ij}^D}\cl_2\psi_{ij}^D, -2\bi\sin(kl/2)\cl_2\log r \Big\}$
  for $(m,n)\in(\mathbb{Z}\times\mathbb{N})^*$ and
  $(i,j)\in\mathbb{Z}\times\mathbb{N}^*$. The above
  TE, TM and TEM modes can be reproduced as well.
\end{myremark}

\medskip

Finally, we use the above waveguide modes to construct solutions to the problem (AHP). 
Let ${\bf F}\in \widetilde{H^{-1/2}}(\Dv,A^h)$, we expand it as
\begin{align}
  \label{eq:tE:as}
  \bf F =& \sum_{(m,n)\in(\mathbb{Z}\times\mathbb{N})^*}d^{TE}_{mn}(\nu\times \bE_{mn}^{TE}|_{A^h}) + \sum_{(i,j)\in\mathbb{Z}\times\mathbb{N}^*} d^{TM}_{ij}(\nu\times \bE_{ij}^{TM}|_{A^h}) \nonumber\\
  &+ d^{TEM}(\nu\times \bE_{E}^{TEM}|_{A^h})\in \widetilde{H^{-1/2}}(\Dv,A^h),
\end{align}
with the Fourier coefficients $\{d_{mn}^{TE},d_{ij}^{TM},d^{TEM}\}$.  Then
it follows from Theorem~\ref{thm:wp:hlp} that the unique solution of the boundary value problem is
\begin{align}
  \label{eq:rep:bE}
  \bE =& \sum_{(m,n)\in(\mathbb{Z}\times\mathbb{N})^*}d^{TE}_{mn}\bE_{mn}^{TE} + \sum_{(i,j)\in\mathbb{Z}\times\mathbb{N}^*} d^{TM}_{ij}\bE_{ij}^{TM}+ d^{TEM}\bE_{E}^{TEM}\in [L^2(G^h)]^3,\\
  \label{eq:rep:bH}
  \bH =& \sum_{(m,n)\in(\mathbb{Z}\times\mathbb{N})^*}d^{TE}_{mn}\bH_{mn}^{TE} + \sum_{(i,j)\in\mathbb{Z}\times\mathbb{N}^*} d^{TM}_{ij}\bH_{ij}^{TM}+ d^{TEM}\bH_{E}^{TEM}\in [L^2(G^h)]^3,
\end{align}
where the modes $\bE_{mn}^{TE}$, $\bH_{mn}^{TE}$, $\bE_{ij}^{TM}$, $\bH_{ij}^{TM}$,
$\bE_{E}^{TEM}$, $\bE_{H}^{TEM}$ are defined in \eqref{eq:E_TE_modes}-\eqref{eq:H_TEM_modes}.
We have
\begin{align*}
  ||\bE||_{[L^2(G^h)]^3}^2 =& \sum_{(m,n)\in(\mathbb{Z}\times\mathbb{N})^*}|d^{TE}_{mn}|^2||\bE_{mn}^{TE}||_{[L^2(G^h)]^3}^2 + \sum_{(i,j)\in\mathbb{Z}\times\mathbb{N}^*} |d^{TM}_{ij}|^2||\bE_{ij}^{TM}||_{[L^2(G^h)]^3}^2\\
  &+ |d^{TEM}|^2||\bE_{E}^{TEM}||_{[L^2(G^h)]^3}^2\\
  =&\sum_{(m,n)\in(\mathbb{Z}\times\mathbb{N})^*}|d^{TE}_{mn}|^2\frac{\lambda_{mn}^{N}}{|s_{mn}^N|}|s_{mn}^Nl + \sin(s_{mn}^Nl)| +\sum_{(i,j)\in\mathbb{Z}\times\mathbb{N}^*}|d^{TM}_{ij}|^2\frac{\lambda_{ij}^{D}}{|s_{ij}^D|}|s_{ij}^Dl + \sin(s_{ij}^Dl)|\\
  &+\sum_{(i,j)\in\mathbb{Z}\times\mathbb{N}^*}|d^{TM}_{ij}|^2\frac{[\lambda_{ij}^{D}]^2}{|s_{ij}^D|^3}|s_{ij}^Dl - \sin(s_{ij}^Dl)|+|d^{TEM}|^22\pi \log(1+h)|kl + \sin(kl)|<\infty,
\end{align*}
and
\begin{align*}
  ||\bH||_{[L^2(G^h)]^3}^2 =& \sum_{(m,n)\in(\mathbb{Z}\times\mathbb{N})^*}|d^{TE}_{mn}|^2||\bH_{mn}^{TE}||_{[L^2(G^h)]^3}^2 + \sum_{(i,j)\in\mathbb{Z}\times\mathbb{N}^*} |d^{TM}_{ij}|^2||\bH_{ij}^{TM}||_{[L^2(G^h)]^3}^2\\
  &+ |d^{TEM}|^2||\bH_{E}^{TEM}||_{[L^2(G^h)]^3}^2\\
  =&\sum_{(m,n)\in(\mathbb{Z}\times\mathbb{N})^*}|d^{TE}_{mn}|^2\frac{|s_{mn}^N|^2\lambda_{mn}^{N}}{k^2|s_{mn}^N|}|s_{mn}^Nl - \sin(s_{mn}^Nl)|\\
  &+\sum_{(m,n)\in(\mathbb{Z}\times\mathbb{N})^*}|d^{TE}_{mn}|^2\frac{|\lambda_{mn}^{N}|^2}{k^2|s_{mn}^N|}|s_{mn}^Nl +\sin(s_{mn}^Nl)|\\
  &+\sum_{(i,j)\in\mathbb{Z}\times\mathbb{N}^*}|d^{TM}_{ij}|^2\frac{k^2\lambda_{ij}^{D}}{|s_{ij}^D|^3}|s_{ij}^Dl -\sin(s_{ij}^Dl)|+|d^{TEM}|^22\pi \log(1+h)|kl -\sin(kl)|<+\infty.
\end{align*}

By Lemmas~\ref{lem:dir:eig} and ~\ref{lem:neu:eig},
$\lambda_{mn}^N,\lambda_{mn}^D\to+\infty$ as $m^2+n^2\to\infty$, so there holds
\[
  |s^o_{mn}|\eqsim \sqrt{\lambda_{mn}^o}, |s_{mn}^{o}l\pm \sin(s_{mn}^ol)|\eqsim
  |2\pm\sin(s_{mn}^{o}l)|,\quad{\rm for}\quad o=N,D,
\]
where $2$ is introduced to ensure that $|2\pm \sin(s_{mn}^o)l)|\geq 1$.
In summary, we have the following proposition.

\begin{myprop}
Let $E, H$ be defined as in 
(\ref{eq:rep:bE})-(\ref{eq:rep:bH}).
Then $||\bE||_{[L^2(G^h)]^3}^2<\infty$ and $||\bH||_{[L^2(G^h)]^3}^2<\infty$ if and only if
\begin{equation}
  \label{eq:cond:coef}
  \left\{  
    \begin{array}{l}
    \{c_{mn}^{TE}:=d_{mn}^{TE}(\lambda_{mn}^N)^{3/4}|2+\sin(s_{mn}^Nl)|^{1/2}\}_{(m,n)\in(\mathbb{Z}\times\mathbb{N})^*}\in \ell^2,\\
     \{c_{ij}^{TM}:=d_{ij}^{TM}(\lambda_{ij}^D)^{1/4}|2+\sin(s_{ij}^Dl)|^{1/2}\}_{(i,j)\in \mathbb{Z}\times\mathbb{N}^*}\in \ell^2,\\
      |d^{TEM}|<\infty,
    \end{array}
\right.
\end{equation}
where $\ell^2$ denotes the space of square-summable sequences. 
On the other hand, for any Fourier coefficients
$\{d_{mn}^{TE},d_{ij}^{TM},d^{TEM}\}$ satisfying (\ref{eq:cond:coef}),  (\ref{eq:rep:bE}) and (\ref{eq:rep:bH}) provide the unique
solution to (AHP) in $H(\cl, G^{h}_+)$ with $\nu\times \bE|_{A^h}\in\widetilde{H^{-1/2}}(\Dv,A^h)$.
\end{myprop}

\begin{myremark}
Unless otherwise stated, here and thereafter, the $\ell^2$ sequence with two indices is arranged in the usual dictionary order. 
\end{myremark}

\begin{myremark}
Transforming the sequence $\{d_{mn}^{TE},d_{ij}^{TM}\}$ to an $\ell^2$
sequence $\{c_{mn}^{TE},c_{ij}^{TM}\}$ balances the magnitudes of the TE, TM, and TEM modes in the hole, which is essential in solving the eigenvalue problem formulated as an infinite-dimensional (INF) linear system by the mode matching method in the next section. As we shall see below, such a transformation eases the analysis of the mapping property of the related  INF coefficient matrix and the reduction of the INF system into finite-dimensional ones. 
\end{myremark}

\section{Quantitative analysis of scattering resonances}
In this section, we  quantitatively characterizes the resonances for the scattering problem (E) in a bounded domain over the complex plane. These resonances are complex values of $k$ such that the homogeneous problem \eqref{eq:problemE_1}-\eqref{eq:problemE_4} 
with $\bE^{\rm inc} = \bH^{\rm inc}=0$ attains nontrivial solutions.

\subsection{A vectorial mode matching formulation}
We first develop a vectorial
analogy of the mode matching method originally proposed in \cite{zholu21,luwanzho21} to reformulate the scattering problem (E) with trivial incident field. 
Before proceeding, we introduce the following bilinear form over
$H^{-1/2}(\Dv,A^h)\times\widetilde{H}^{-1/2}(\Dv,A^h)$ (see\cite[P. 306]{kirhet15}):
\[
  \langle\bF,\bG\rangle =: \langle \bF,\nu\times \bG  \rangle_{A^h},
\]
where $\nu\times\bG=-(\bG\times\nu)\in \widetilde{H^{-1/2}}(\Cl,A^h)$, and
$\langle \cdot,\cdot \rangle_{A^h}$ represents the duality pair between $H^{-1/2}(\Dv,A^h)$
and $\widetilde{H^{-1/2}}(\Cl,A^h)$. 
Let ${\cal S}_{k}$ be the following single-layer operator
\begin{equation}
  \label{eq:Sh}
  {\cal S}_{k}[\phi](x) =\int_{A^{h}} \Phi_k(x;y)\phi(y)\, dS(y),\quad x\in A^{h}.
\end{equation}
Then ${\cal S}_{k}$ is bounded from $\widetilde{H^{-1/2}}(A^h)$
to $H^{1/2}(A^h)$. The following holds for the T2T map ${\cal L}_k$.
\begin{mylemma}[\cite{kirhet15}, Lemma 5.61]
  For any $\bF,\bG\in \widetilde{H}^{-1/2}(\Dv,A^h)$,
  \begin{align}
    \langle{\cal L}_k[\bF],\bG  \rangle =& \langle {\cal L}_k[\bG],\bF \rangle,\\
    \label{eq:LkSk}
    \langle{\cal L}_k[\bF],\bG  \rangle =& -\langle\Dv~\bG, {\cal S}_k[\Dv~\bF] \rangle_{A^h} + k^2\langle\bG, {\cal S}_k[\bF]  \rangle_{A^h},
  \end{align}
  where ${\cal S}_k[\bF]$ is taken componentwisely, and it belongs to $H^{-1/2}(\Cl, A^h)$.
\end{mylemma}

Now, from the integral equation formulation \eqref{eq:EH:hsp} and the tangential traces of $\bE$ and $\bH$ (\ref{eq:rep:bE}) (\ref{eq:rep:bH}) over the annular aperture $A^h$, when $\bE^{\rm inc} = \bH^{\rm inc}=0$, the homogeneous problem \eqref{eq:problemE_1}-\eqref{eq:problemE_4} can be formulated as the following system over the aperture $A^h$:
\begin{align}
  \label{eq:t2t}
  \nu\times \bH|_{A^h} =& \frac{-2}{\bi k}{\cal L}_k[\nu\times \bE|_{A^h}],\\
  \label{eq:rep:tE}
  \nu\times \bE|_{A^h} =& \sum_{(m,n)\in(\mathbb{Z}\times\mathbb{N})^*}\frac{c^{TE}_{mn} \cdot 2\cos(s_{mn}^Nl/2) }{(\lambda_{mn}^N)^{3/4}|2+\sin(s_{mn}^Nl)|^{1/2}}{\cal R}\cl_2\ \psi_{mn}^N\nonumber\\
  &+ \sum_{(i,j)\in\mathbb{Z}\times\mathbb{N}^*} \frac{c^{TM}_{ij} \cdot 2\cos(s_{ij}^Dl/2) }{(\lambda_{ij}^{D})^{1/4}|2+\sin(s_{ij}^Dl)|^{1/2}}{\cal R}\nabla_2\psi_{ij}^{D} \nonumber\\
  &+ d^{TEM} \cdot 2\cos(kl/2){\cal R}\nabla_2\log r,\\
  \label{eq:rep:tH}
  \nu\times\bH|_{A^h}=&\sum_{(m,n)\in(\mathbb{Z}\times\mathbb{N})^*}\frac{c^{TE}_{mn} \cdot 2\bi s_{mn}^N\sin(s_{mn}^Nl/2) }{k(\lambda_{mn}^N)^{3/4}|2+\sin(s_{mn}^Nl)|^{1/2}}\cl_2\ \psi_{mn}^N\nonumber\\
  &+ \sum_{(i,j)\in\mathbb{Z}\times\mathbb{N}^*} \frac{c^{TM}_{ij} \cdot 2\bi k\sin(s_{ij}^Dl/2) }{s_{ij}^D(\lambda_{ij}^{D})^{1/4}|2+\sin(s_{ij}^Dl)|^{1/2}}\nabla_2\psi_{ij}^D\nonumber\\
  &+ d^{TEM} \cdot 2\bi \sin(kl/2)\nabla_2\log r. 
\end{align}
At a resonance $k$, there exist nontrivial solutions $\{c^{TE}_{mn}, c^{TM}_{mn},  d^{TEM}\}$ for the above system.

Using the completeness of the basis given in (\ref{eq:bs:tot}), the integral equation (\ref{eq:t2t}) is equivalent to the following system:
\begin{align}
  \langle \nu\times \bH|_{A^h}, {\cal R}\cl_2 \overline{\psi_{m'n'}^{N}} \rangle =& \frac{-2}{\bi k} \langle{\cal L}_k[\nu\times \bE|_{A^h}], {\cal R}\cl_2 \overline{\psi_{m'n'}^{N}} \rangle,\quad (m',n')\in(\mathbb{Z}\times\mathbb{N})^*, \label{eq:system1} \\
  \langle \nu\times \bH|_{A^h}, {\cal R}\nabla_2 \overline{\psi_{i'j'}^{D}} \rangle =& \frac{-2}{\bi k} \langle{\cal L}_k[\nu\times \bE|_{A^h}], {\cal R}\nabla_2 \overline{\psi_{i'j'}^{D}} \rangle,\quad(i',j')\in\mathbb{Z}\times\mathbb{N}^*, \label{eq:system2} \\
  \langle \nu\times \bH|_{A^h}, {\cal R}\nabla_2 \log r \rangle =& \frac{-2}{\bi k} \langle{\cal L}_k[\nu\times \bE|_{A^h}], {\cal R}\nabla_2 \log r \rangle, \label{eq:system3}
\end{align}
where the overline represents the complex conjugate. 
Using the expansions (\ref{eq:rep:tE}), (\ref{eq:rep:tH}), and the following identities
\begin{align}
 \langle \cl_2 \psi_{mn}^N,{\cal R}\cl_2 \overline{\psi_{m'n'}^{N}}  \rangle  =& -\lambda_{mn}^{N}\delta_{mm'}\delta_{nn'}, \label{eq:identity1} \\
 \langle \nabla_2 \psi_{ij}^D,{\cal R}\nabla_2 \overline{\psi_{i'j'}^{D}}  \rangle  =& -\lambda_{ij}^{D}\delta_{ii'}\delta_{jj'}, \label{eq:identity2} \\
 \langle \nabla_2 \log r,{\cal R}\nabla_2 \log r  \rangle  =& -2\pi\log(1+h), \label{eq:identity3} 
\end{align}
where $\delta_{m0}$ is the Kronecker delta function, the system \eqref{eq:system1}-\eqref{eq:system3} can be rewritten as an equation of INF matrices
and INF vectors:
\begin{align}
  \label{eq:INF:sys}
  &\left[
  \begin{array}{lll}
    \bS^{TE}\bD^{TE} & & \\
     & \bD^{TM} & \\
     &  & D^{TEM}\\
  \end{array}
  \right]\left[
  \begin{array}{l}
    \bc^{TE}\\
    \bc^{TM}\\
    d^{TEM}\\
  \end{array}
  \right] \nonumber\\
  =&  \left[
  \begin{array}{lll}
    \bA^{TE,TE} & \bA^{TE,TM} & \bC^{TE,TEM}\\
    \bA^{TM,TE} & \bA^{TM,TM} & \bC^{TM,TEM}\\
    \bR^{TEM,TE} & \bR^{TEM,TM} & A^{TEM,TEM}\\
  \end{array}
  \right] \left[
  \begin{array}{l}
    \bc^{TE}\\
    \bc^{TM}\\
    d^{TEM}\\
  \end{array}
  \right].
\end{align}
In the above, the unknown coefficients are given by 
the two INF column vectors $\bc^{TE} =
[c_{m'n'}^{TE}]_{(m',n')\in(\mathbb{Z}\times\mathbb{N})^*}$ and $\bc^{TM} =
[c_{i'j'}^{TM}]_{(i',j')\in\mathbb{Z}\times\mathbb{N}^*}$, and a complex number $d^{TEM}$. They represent the Fourier coefficients of the TE, TM, and TEM modes respectively in \eqref{eq:rep:tE} and \eqref{eq:rep:tH}. 

On the left side of the system, $D^{TEM} =\sin(kl/2)$, and
the three  INF diagonal matrices are given by
\begin{align*}
  \bS^{TE} = &{\rm Diag}\{s_{mn}^N\},\\
  \bD^{TE} =& {\rm Diag }\{\sin(s_{mn}^Nl/2)|2+\sin(s_{mn}^Nl)|^{-1/2}\},\\
  \bD^{TM} =& {\rm Diag }\{\sin(s_{ij}^Dl/2)|2+\sin(s_{ij}^Dl)|^{-1/2}\}.
\end{align*}
The elements in the matrices are obtained from the field representation \eqref{eq:rep:tH} and the identities \eqref{eq:identity1}-\eqref{eq:identity3}. We use the superscripts to denote the contribution of each type of mode to the matrices.
On the right side of the system, the four INF matrices are 
\begin{align}
  \bA^{TE,TE} =& \left[\frac{-2\cos(s_{mn}^Nl/2)|2+\sin(s_{mn}^Nl)|^{-1/2}}{(\lambda_{m'n'}^{N})^{1/4}(\lambda_{mn}^N)^{3/4}}\langle {\cal L}_k{\cal R}\cl_2\psi_{mn}^{N},{\cal R}\cl_2\overline{\psi_{m'n'}^{N}} \rangle\right],\nonumber\\
  \bA^{TE,TM} =& \left[\frac{-2\cos(s_{ij}^Dl/2)|2+\sin(s_{ij}^Dl)|^{-1/2}}{(\lambda_{m'n'}^{N})^{1/4}(\lambda_{ij}^D)^{1/4}}\langle {\cal L}_k{\cal R}\nabla_2\psi_{ij}^{D},{\cal R}\cl_2\overline{\psi_{m'n'}^{N}} \rangle\right],\nonumber\\
  \bA^{TM,TE} =& \left[\frac{-2s_{i'j'}^D\cos(s_{mn}^Nl/2)|2+\sin(s_{mn}^Nl)|^{-1/2}}{k^2(\lambda_{i'j'}^D)^{3/4}(\lambda_{mn}^N)^{3/4}}\langle {\cal L}_k{\cal R}\cl_2\psi_{mn}^{N},{\cal R}\nabla_2\overline{\psi_{i'j'}^{D}} \rangle\right],\nonumber\\
  \bA^{TM,TM} =& \left[\frac{-2s_{i'j'}^D\cos(s_{ij}^Dl/2)|2+\sin(s_{ij}^Dl)|^{-1/2}}{k^2(\lambda_{i'j'}^D)^{3/4}(\lambda_{ij}^D)^{1/4}}\langle {\cal L}_k{\cal R}\nabla_2\psi_{ij}^{D},{\cal R}\nabla_2\overline{\psi_{i'j'}^{D}} \rangle\right]. \nonumber
\end{align}
The two INF column vectors $\bC^{TE,TEM}$ and $\bC^{TM,TEM}$, and two INF row vectors $\bR^{TEM,TE}$, and $\bR^{TEM,TM}$ are
\begin{align}
  \bC^{TE,TEM} =& \left[\frac{-2\cos(kl/2)}{(\lambda_{m'n'}^{N})^{1/4}}\langle {\cal L}_k{\cal R}\nabla_2\log r,{\cal R}\cl_2\overline{\psi_{m'n'}^{N}} \rangle\right],\nonumber\\
  \bC^{TM,TEM} =& \left[\frac{-2s_{i'j'}^D\cos(kl/2)}{k^2(\lambda_{i'j'}^D)^{3/4}}\langle {\cal L}_k{\cal R}\nabla_2\log r,{\cal R}\nabla_2\overline{\psi_{i'j'}^{D}} \rangle\right],\nonumber\\
  \bR^{TEM,TE} =& \left[\frac{-\cos(s_{mn}^Nl/2)|2+\sin(s_{mn}^Nl)|^{1/2}}{(\lambda_{mn}^N)^{3/4}\pi k\log(1+h)}\langle {\cal L}_k{\cal R}\cl_2\psi_{mn}^{N},{\cal R}\nabla_2\log r \rangle\right],\nonumber\\
  \bR^{TEM,TM} =& \left[\frac{-\cos(s_{ij}^Dl/2)|2+\sin(s_{ij}^Dl)|^{1/2}}{(\lambda_{ij}^D)^{1/4}\pi k\log(1+h)}\langle {\cal L}_k{\cal R}\nabla_2\psi_{ij}^{D},{\cal R}\nabla_2\log r \rangle\right],\nonumber
\end{align}
and the scalar
\begin{align}
  A^{TEM,TEM} =\frac{-\cos(kl/2)}{\pi k\log(1+h)}\langle {\cal L}_k{\cal R}\nabla_2\log r,{\cal R}\nabla_2\log r \rangle.\nonumber
\end{align}
The elements in the matrices and vectors are obtained from using the expansion \eqref{eq:rep:tE} for the systems \eqref{eq:system1}-\eqref{eq:system3} and the identities \eqref{eq:identity1}-\eqref{eq:identity3}. Each pair of superscript for the matrix/vector denotes the interaction of two modes after applying the operator ${\cal L}_k$ to one mode.
We set the following rules for the indices of the elements of the INF matrices/vectors:
\begin{itemize}
\item[(1).] $(m,n)$ and $(m',n')$ range over
  $(\mathbb{Z}\times\mathbb{N})^*$;
\item[(2).] $(i,j)$ and $(i',j')$ range over
$\mathbb{Z}\times\mathbb{N}^*$;
\item[(3).] the index $(m,n)$ or $(i,j)$ is the column index
  of the matrix, while the prime index $(m',n')$ or $(i',j')$ is the row index of the matrix.
\item[(4).] The columns (and rows) of each INF matrix are arranged in the dictionary order.
\end{itemize}
The product of the block INF matrix and the block INF vector in
 (\ref{eq:INF:sys}) is well-defined by the usual matrix-vector product.

\subsection{Resonances for Problem (E)}
We are ready to analyze the resonances for the scattering problem (E), which are the characteristic values of the system  (\ref{eq:INF:sys}). We shall follow the avenues described below to derive their asymptotic expansions:
\begin{enumerate}
\item[(1).] First, we decompose the whole system  (\ref{eq:INF:sys}) into a sequence of subsystems \eqref{eq:eig:m} with different angular momentum $m\in\mathbb{Z}$. 
\item[(2).]  We further reduce each subsystem \eqref{eq:eig:m} to a nonlinear characteristic equation \eqref{eq:single} by projecting the solution onto the dominant resonant mode. Such a characteristic equation is called resonance condition. To this end, we
estimate the contribution from the modes that are orthogonal to the resonant modes in each subsystem, which is accomplished by the asymptotic analysis of each matrix element with respect to the parameter $h$ and the key estimates are provided in Lemma \ref{lem:dm'm:h}.
\item[(3).]  Finally, we investigate the resonance condition \eqref{eq:single} and analyze its roots to obtain the asymptotic expansions of resonances. The main results for the resonances are summarized in Theorems~\ref{thm:even:res} and \ref{thm:odd:res}.
\end{enumerate}

\subsubsection{Subsystem for each angular momentum}
For a function depending on the angle $\theta$, we use $\Theta(f)$ to denote its angular momentum so that the $\theta$-dependence of the function is given by $e^{\bi \Theta(f)\theta}$. For example, for $\psi_{ij}^D$ and $\psi_{mn}^N$ defined \eqref{eq:phimnD} and \eqref{eq:psimnN}, there holds $\Theta(\psi_{mn}^N)=m$
and $\Theta(\psi_{ij}^D)=i$. On the other hand, $\Theta(\log r)=0$. We have the following orthogonality relation for two basis functions with different momenta.
\begin{mylemma}
\label{lem:orth:theta}
  For any $f,g\in\{\psi_{mn}^N,\psi_{ij}^D,\log
  r\}_{(m,n,i,j)\in(\mathbb{Z}\times\mathbb{N})^*\times \mathbb{Z}\times\mathbb{N}^*}$ with
$\Theta(f)\neq \Theta(g)$, there holds
  \begin{align}
    \langle {\cal L}_k {\cal R}{\rm Op}_1 [f],{\cal R}{\rm Op}_2 \overline{[g]} \rangle =& 0,
  \end{align}
  where ${\rm Op}_j$ represents one of the two operators $\{\cl_2,\nabla_2\}$,
  for $j=1,2$.
  \begin{proof}
    We only show the proof when $\Op_1 =\cl_2$, $\Op_2 =
    \nabla_2$, $f=\psi_{mn}^N$ and $g=\psi_{ij}^D$ with $m\neq i$. For
    simplicity, let $f(r,\theta)=f_n(r)e^{\bi m\theta}$ and
    $g(r',\theta')=g_j(r')e^{\bi i\theta'}$, where both $f_n$ and
    $g_j$ are real. A direction calculation gives
    \begin{align*}
      \nabla_2 f \cdot \cl_2'\bar{g} =& (f_n'(r)e^{\bi m\theta}\hat{r} +\bi m
      f_n(r)e^{\bi m\theta}\hat{\theta})\cdot (-g_j'(r')e^{-\bi i\theta'}\hat{\theta}'+\bi i
      f_n(r)e^{-\bi i\theta'}\hat{r}')\\
      =& [h^1_{nj}(r,r')\cos(\theta-\theta') + h^2_{nj}(r,r')\sin(\theta-\theta')]e^{\bi m\theta-\bi i\theta'},
    \end{align*}
    where $\hat{\theta}$ and $\hat{r}$ are the polar unit vectors, and $h^{o}_{nj},o=1,2$ are uniquely determined
    from $f_n$, $g_j$ and their first-order derivatives. Thus by
    (\ref{eq:LkSk}),
    \begin{align*}
      &\langle {\cal L}_k {\cal R}\cl_2 [f],{\cal R}\nabla_2' \overline{[g]} \rangle \\
      =& -k^2\langle \nabla_2[f], {\cal S}_k[\cl_2'\overline{g}]  \rangle_{A^h}\\
      =&-k^2\int_{0}^{2\pi}e^{\bi (m-i)\theta'}d\theta' \int_{0}^{2\pi}d\theta\int_{[a,a(1+h)^2]}\frac{e^{\bi k \sqrt{r^2+r'^2-2rr'\cos\theta}}[h_{nj}^1\cos\theta+h_{nj}^2\sin\theta]}{4\pi|r^2+r'^2-2rr'\cos\theta|} e^{\bi m\theta}drdr'\\
      =& \, 0.
    \end{align*}
    The proof for the other cases are similar.
  \end{proof}
\end{mylemma}
Using the above lemma, the full system (\ref{eq:INF:sys}) can be decoupled into a
sequence of subproblems, where the elements in each subsystem attain the same angular dependence $e^{\bi m \theta}.$
More specifically, for each $m\in\mathbb{Z}$, we have
\begin{align}
  \label{eq:eig:m}
  \left[
  \begin{array}{lll}
     D_m &  & \\
    & \bI & \\
     & & \bI \\
  \end{array}
  \right]\left[
  \begin{array}{l}
    d_m\\
    \bc_m^{TE}\\
    \bc_m^{TM}\\
  \end{array}
  \right] = \left[
  \begin{array}{lll}
    A_{mm} & \bR_m^{TE} & \bR_m^{TM}\\
    \bC_m^{TE} & \bB_m^{TE,TE} & \bB_m^{TE,TM}  \\
    \bC_m^{TM} & \bB_m^{TM,TE} & \bB_m^{TM,TM} \\
  \end{array}
  \right] \left[
  \begin{array}{l}
    d_m\\
    \bc_m^{TE}\\
    \bc_m^{TM}\\
  \end{array}
  \right].
\end{align}
In the above, the unknown coefficients for each $m$ are
\begin{align*}
d_m= \left\{
  \begin{array}{lc}
    d^{TEM},& m=0;\\
    d_{m0}^{TE},& m\neq 0,\\
  \end{array}
  \right.\quad \bc_m^{TE}=[D_{mn'}^{TE}c_{mn'}^{TE}]_{n'\in\mathbb{N}^*},\quad
\bc_m^{TM}=[D_{mj'}^{TM}c_{mj'}^{TM}]_{j'\in\mathbb{N}^*}.
\end{align*}
${\bI}$ denotes the INF identity matrix on $\ell^2$ such that $\bI\bc_m^j=\bc_m^j, j\in\{TE,TM\}$. The scalar
$D_0=\sin (kl/2)$ and $D_m=s_{m0}^N\sin (s_{m0}^Nl/2)$ for $m\neq 0$. The $3\times 3$ block matrices, relating to the $3\times
3$ block matrices on the right-side in (\ref{eq:INF:sys}) for each $m$, are given as follows:
\begin{align*}
  & A_{mm} := \left\{
  \begin{array}{lc}
    A^{TEM,TEM},& m=0;\\
    |2+\sin(s_{m0}^Nl)|^{1/2}A_{m0,m0}^{TE,TE},& m\neq 0,\\
  \end{array}
  \right.\\
  \\
  & \bR_m^{TE}:=[R_{n;m}^{TE}]_{n\in\mathbb{N^*}}=\left\{
  \begin{array}{lc}
    \ \left[R_{mn}^{TEM,TE}(D^{TE}_{mn})^{-1}\right]_{n\in\mathbb{N}^*},& m=0;\\ \\
    \ \left[(\lambda_{m0}^{N})^{-3/4}A^{TE,TE}_{m0,mn}(D^{TE}_{mn})^{-1}\right]_{n\in\mathbb{N}^*}, & m\neq 0,\\
  \end{array} \right.\\ \\
   & \bR_m^{TM}:=\left[R_{j;m}^{TM}\right]_{j\in\mathbb{N^*}}=\left\{
  \begin{array}{lc}
    \ \left[R_{mj}^{TEM,TM}(D^{TM}_{mj})^{-1}\right]_{j\in\mathbb{N}^*},& m=0;\\ \\
    \ \left[(\lambda_{m0}^{N})^{-3/4}A^{TE,TM}_{m0,mj}(D^{TM}_{mj})^{-1}\right]_{j\in\mathbb{N}^*}, & m\neq 0,\\
  \end{array} \right.\\ \\
  & \bC_m^{TE}:=\left[C_{n';m}^{TE}\right]_{n'\in\mathbb{N^*}}=\left\{
  \begin{array}{lc}
    \ \left[(s_{mn'}^{N})^{-1}C_{mn'}^{TE,TEM}\right]_{n'\in\mathbb{N}^*},& m=0;\\ \\
    \ \left[(s_{mn'}^{N})^{-1}A^{TE,TE}_{mn',m0}(\lambda_{m0}^N)^{3/4}|2+\sin(s_{m0}^Nl)|^{1/2}\right]_{n'\in\mathbb{N}^*}, & m\neq 0,\\
  \end{array} \right.\\ \\
   & \bC_m^{TM}:=\left[C_{j';m}^{TM}\right]_{j'\in\mathbb{N^*}}=\left\{
  \begin{array}{lc}
    \ \left[C_{mj'}^{TM,TEM}\right]_{j'\in\mathbb{N}^*},& m=0;\\ \\
    \ \left[A^{TM,TE}_{mj',m0}(\lambda_{m0}^N)^{3/4}|2+\sin(s_{m0}^Nl)|^{1/2}\right]_{j'\in\mathbb{N}^*}, & m\neq 0,\\
  \end{array} \right. \\ \\
  & \bB_m^{TE,TE} := \left[B^{TE,TE}_{n'n;m}=(s^N_{mn'})^{-1}A^{TE,TE}_{mn',mn}(D^{TE}_{mn})^{-1}\right]_{n',n\in\mathbb{N}^*},\\ \\
  & \bB_m^{TM,TE} :=\left[B^{TM,TE}_{j'n;m}=A^{TM,TE}_{mj',mn}(D^{TE}_{mn})^{-1}\right]_{j',n\in\mathbb{N}^*},\\ \\
  & \bB_m^{TE,TM} := \left[B^{TE,TM}_{n'j;m}=(s^{N}_{mn'})^{-1}A^{TE,TM}_{mn',mj}(D^{TM}_{mj})^{-1}\right]_{n',j\in\mathbb{N}^*},\\ \\
  & \bB_m^{TM,TM} :=\left[B^{TM,TM}_{j'j;m}=A^{TM,TM}_{mj',mj}(D^{TM}_{mj})^{-1}\right]_{j',j\in\mathbb{N}^*}.
\end{align*}
In the above, $D^{TE}_{mn'}$ denotes the $mn'$-th diagonal element of $\bD^{TE}$ in \eqref{eq:INF:sys}, and 
$A^{TE,TE}_{mn',mn}$ denotes the $mn'$-th row, $mn$-th column element of
$\bA^{TE,TE}$ in \eqref{eq:INF:sys}, etc. The square bracket $[\cdot]$ represents an INF matrix, an INF row vector or an INF column vector, with the subscript given by the following:
\begin{itemize}
\item[1.] The subscript $n\in\mathbb{N}^*$ (or $j\in\mathbb{N}^*$) represents the column index $n$ (or $j$) in a row vector;
\item[2.] The subscript $n'\in\mathbb{N}^*$ or $j'\in\mathbb{N}^*$ with a prime represents row index $n'$ (or $j'$) a column vector;
\item[3.] The subscript $n',n\in\mathbb{N}^*$ represents the column index $n'$ and row index $n$ for a matrix.
\end{itemize}
Consequently, to solve for \eqref{eq:INF:sys} it is equivalent to solve for $k\in{\cal B}$ such that, for each $m$, the system (\ref{eq:eig:m})
attains nonzero $\ell^2$-sequences $\{d_m, c_m^{TE},c_m^{TM}\}$.

\subsubsection{Characteristic equation and resonance condition}
To proceed, we transform each system (\ref{eq:eig:m}) to an equivalent characteristic equation. To this end, we first analyze the matrix elements in (\ref{eq:eig:m}) for $h\ll 1$. Let ${\cal S}_0$ be a single-layer potential over the interval
$(0,1)$ given by
\begin{align}
  ({\cal S}_0 [\phi])(r):= 2\int_{0}^{1}\frac{1}{2\pi}\log\frac{1}{|r-r'|} \phi(r')ds(r'),
\end{align}
wherein the kernel function is the fundamental solution of
the 2D Laplacian. It is known that ${\cal S}_0$ is bounded from
$\widetilde{H^{-1/2}}(0,1)$ to $H^{1/2}(0,1) =
(\widetilde{H^{-1/2}}(0,1))'$ \cite[Lem. 2.1.2]{luwanzho21}. Let
\begin{equation}
  \label{eq:phi}
   \phi_n(r) = \begin{cases}
    \frac{1}{\sqrt{2}}\cos(n\pi r),& n\in\mathbb{N}^*,\\
    1,& n=0.
    \end{cases}
\end{equation}
Then, $\{\phi_n\}_{n\in\mathbb{N}}$ forms an orthonormal basis of the space
$L^2(0,1)$. We
equip $H^{1/2}(0,1)$ with the following norm:
\[
  ||f||_{H^{1/2}(0,1)}^2:=\sum_{n=0}^{\infty}(1+n^2)^{1/2}|(f,\phi_n)_{L^2(0,1)}|^2,
\]
and $\widetilde{H^{-1/2}}(0,1)$ with the norm
\[
  ||f||_{\widetilde{H^{-1/2}}(0,1)}^2:=\sum_{n=0}^{\infty}(1+n^2)^{-1/2}|\langle
  f,\phi_n\rangle_{(0,1)} |^2,
\]
where $\langle \cdot,\cdot \rangle_{(0,1)}$ indicates the duality pair between
$\widetilde{H^{-1/2}}(0,1)$ and $H^{1/2}(0,1)$. The estimations of the matrix elements in \eqref{eq:eig:m} are given in the following lemma:
\begin{mylemma}
  \label{lem:dm'm:h}
  Let $h\ll 1$ and $k\in{\cal B}$. For each $m\in\mathbb{Z}$, the following hold:
  \begin{itemize}
  \item[(i).] The element $B_{n'n;m}^{TE,TE}$ in the matrix $\bB_m^{TE,TE}$ attains the following asymptotic expansions:
      \begin{align}
        \label{eq:n'n:TETE}
        B_{n'n;m}^{TE,TE} = -2({\cal S}_0[(n'\pi)^{1/2}\phi_{n'}],(n\pi)^{1/2}\phi_n)_{L^2(0,1)} + {\cal O}(h)\epsilon_{n'n;m}^{TE,TE},
      \end{align}
      where the INF matrix
      $\{\epsilon_{n'n;m}^{TE,TE}\}_{n,n'=1}^{\infty}:\ell^2\to\ell^2$ is uniformly
      bounded for $h\ll1$.
  \item[(ii).] The element $B_{n'j;m}^{TE,TM}$ in the INF matrix $\bB_m^{TE,TM}$ attains the following asymptotic expansions:
      \begin{align}
        \label{eq:n'j:TETM}
        B_{n'j;m}^{TE,TM} = (1-\delta_{m0}){\cal O}(h)\epsilon_{n'j;m}^{TE,TM},
      \end{align}
      where the INF matrix
      $\{\epsilon_{n'n;m}^{TE,TM}\}_{n,n'=1}^{\infty}:\ell^2\to\ell^2$ is uniformly
      bounded for $h\ll1$.
   \item[(iii).] The element $B_{j'n;m}^{TM,TE}$ in the INF matrix $\bB_m^{TM,TE}$ attains the following asymptotic expansions:
      \begin{align}
        \label{eq:j'n:TMTE}
        B_{j'n;m}^{TM,TE} = (1-\delta_{m0}){\cal O}(h)\epsilon_{j'n;m}^{TM,TE},
      \end{align}
      where the INF matrix
      $\{\epsilon_{j'n;m}^{TM,TE}\}_{n,n'=1}^{\infty}:\ell^2\to\ell^2$ is uniformly
      bounded for $h\ll1$.
   \item[(iv).] The element $B_{j'j;m}^{TM,TM}$ in the INF matrix $\bB_m^{TE,TE}$ attains the following asymptotic expansions:
      \begin{align}
        \label{eq:j'j:TMTM}
        B_{j'j;m}^{TM,TM} = -2({\cal S}_0[(n'\pi)^{1/2}\phi_{j'}],(n\pi)^{1/2}\phi_j)_{L^2(0,1)} + {\cal O}(h)\epsilon_{j'j;m}^{TM,TM},
      \end{align}
      where the INF matrix
      $\{\epsilon_{n'n;m}^{TE,TE}\}_{n,n'=1}^{\infty}:\ell^2\to\ell^2$ is
      uniformly bounded for $h\ll1$.

    \item[(v).] The two INF column vectors $\bC_m^{TE}$ and
      $\bC_m^{TM}$ are uniformly bounded in $\ell^2$ as $h\to 0+$.
      For $m\neq 0$, 
      \begin{align}
        \label{eq:Cn'mTE}
        C_{n';m}^{TE} =& -\bi \lambda_{m0}\cos(s_{m0}^Nl/2)h^{1/2}\left[({\cal S}_0[(n'\pi)^{1/2}\phi_{ n' }],\phi_{0})_{L^2(0,1)} + {\cal O}(h\log h)\epsilon_{n';m}^{CTE}  \right],\\
        \label{eq:Cj'mTM}
        C_{j';m}^{TM} =& m\cos(s_{m0}^Nl/2)h^{1/2}\left[({\cal S}_0[\phi_0],(j'\pi)^{1/2}\phi_{j'})_{L^2(0,1)} +{\cal O}(h\log h)\epsilon_{n';m}^{CTM}\right],
      \end{align}
      and for $m=0$,
      \begin{align}
        \label{eq:Cn'mTE:m=0}
        C_{n';m}^{TE} =& 0,\\
        \label{eq:Cj'mTM:m=0}
        C_{j';m}^{TM} =& -\sqrt{2\pi}\bi \cos(kl/2)h\left[({\cal S}_0[\phi_0],(j'\pi)^{1/2}\phi_{j'})_{L^2(0,1)} +{\cal O}(h\log h)\epsilon_{n';m}^{CTM}\right],
      \end{align}
      where the two INF column vectors $\{\epsilon_{n';m}^{CTE}\}$ and
      $\{\epsilon_{j';m}^{CTM}\}$ are uniformly bounded in $\ell^2$ for $h\ll1$.

    \item[(vi).] The two INF row vectors $\bR_m^{TE}$ and
      $\bR_m^{TM}$ are uniformly bounded in $\ell^2$ as $h\to 0+$.
      For $m\neq 0$,
      \begin{align}
        \label{eq:RnmTE}
        R_{n;m}^{TE} =& -\bi h^{1/2}({\cal S}_0[\phi_{0}],(n\pi)^{1/2}\phi_{ n })_{L^2(0,1)} + {\cal O}(h^{3/2} \log h)\{\epsilon_{n;m}^{RTE}\},\\
        \label{eq:RjmTM}
        R_{j;m}^{TM} =& \frac{m}{\lambda_{m0}}k^2h^{1/2}\left[({\cal S}_0[\phi_0],(j\pi)^{1/2}\phi_{j})_{L^2(0,1)} + {\cal O}(h\log h)\{\epsilon_{j;m}^{RTM}\}  \right],
      \end{align}
      and for $m=0$,
      \begin{align}
        \label{eq:RnmTE:m=0}
        R_{n;m}^{TE} =& 0,\\
        \label{eq:RjmTM:m=0}
        R_{j;m}^{TM} =& \frac{\bi k}{\sqrt{2\pi}} \left[({\cal S}_0[(j\pi)^{1/2}\phi_{j}],\phi_0)_{L^2(0,1)} + {\cal O}(h\log h)\{\epsilon_{j;m}^{RTM}\}  \right],
      \end{align}
      where the two INF row vectors $\{\epsilon_{n;m}^{RTE}\}$ and
      $\{\epsilon_{j;m}^{RTM}\}$ in $\ell^2$ are uniformly bounded  for $h\ll1$.
    \item[(vii).] As $h\to 0^+$, for $m\neq 0$,
      \begin{align}
        \label{eq:Amm}
        A_{mm} =& 2\cos(s_{m0}^Nl/2)\lambda_{m0}^N[-h\frac{\log h}{4\pi} + \alpha_m(k)h + \bi \beta_m(k)h]\nonumber\\
                &-\frac{2k^2m^2\cos(s_{m0}^Nl/2)}{\lambda_{m0}}[-h\frac{\log h}{4\pi} + \tilde{\alpha}_m(k)h + \bi \tilde{\beta}_m(k)h ] \nonumber\\
        &+ \cos(s_{m0}^Nl/2){\cal O}(h^2\log h),
      \end{align}
      and for $m=0$,
      \begin{align}
        \label{eq:Amm:m=0}
        A_{mm} =& -\cos(kl/2)k[-h\frac{\log h}{4\pi} + \alpha_1(k)h + \bi \beta_1(k)h + {\cal O}(h^2\log h)],
      \end{align}
      where for $m\in\mathbb{Z}$,
\begin{align}
  \label{eq:alpham:k}
  \alpha_m(k) 
  =&\frac{3}{8\pi}+ \frac{1}{\pi}\int_{0}^{\pi/2}\frac{(\cos(k\sin(\theta))-1)\cos(2m\theta)}{\sin(\theta)}d\theta \nonumber\\
  &+ \frac{1}{4\pi}[\log 2 -\gamma - \psi(|m|+1/2)],\\
  \label{eq:betam:k}
  \beta_m(k) = &\frac{1}{\pi}\int_{0}^{\pi/2}\frac{\sin(k\sin(\theta))\cos(2m\theta)}{\sin(\theta)}d\theta=\frac{1}{2}\int_{0}^{k}J_{2m}(t)dt,\\
  \label{eq:talpham:k}
  \tilde{\alpha}_m(k)=&\frac{\alpha_{m+1}(k)+\alpha_{m-1}(k)}{2},\\
  \label{eq:tbetam:k}
  \tilde{\beta}_m(k)=&\frac{\beta_{m+1}(k)+\beta_{m-1}(k)}{2},
\end{align}
$\gamma$ is Euler's constant, and $\psi$ denotes the logarithmic derivative of
gamma function (c.f. \cite[\S 5.2(i)]{nist10}).
  \end{itemize}
  In the above, the prefactors in the ${\cal O}$-notations depend only on ${\cal
  B}$
  and $m$.
  \begin{proof}
    Details of the proof are presented in Appendix C.
  \end{proof}
\end{mylemma}

Now, we define three INF matrices
\begin{align*}
{\bB}_m:=\left[
  \begin{array}{lll}
    \bB_m^{TE,TE} & \bB_m^{TE,TM}  \\
    \bB_m^{TM,TE} & \bB_m^{TM,TM} \\
  \end{array}
  \right],\ {\bP}:=[p_{n'n}]_{n',n\in\mathbb{N}^*},\ {\bP}_2:=\left[
  \begin{array}{lll}
     \bP &   \\
     & \bP \\
  \end{array}
  \right],
\end{align*}
which are uniformly bounded from $\ell^2$ to $\ell^2$ for $h\ll 1$, and four INF column/row
vectors
\begin{align*}
  \bc_m:=\left[
  \begin{array}{l}
     \bc_m^{TE}\\
    \bc_m^{TM} \\
  \end{array}
  \right],\quad  \bR_m:=\left[
  \begin{array}{ll}
     \bR_m^{TE} & \bR_m^{TM} \\
  \end{array}
  \right],\quad\bC_m:=\left[
  \begin{array}{l}
     \bC_m^{TE}\\
    \bC_m^{TM} \\
  \end{array}
  \right],\quad \bp:=[p_{n'0}]_{n'\in\mathbb{N}^*},
\end{align*}
uniformly bounded in $\ell^2$. In the above, the element 
\[
  p_{n'n} =\begin{cases}
  ({\cal S}_0[(n'\pi)^{1/2}\phi_{n'}],(n\pi)^{1/2}\phi_n)_{L^2(0,1)} &
  n\in\mathbb{N}^*;\\
  ({\cal S}_0[(n'\pi)^{1/2}\phi_{n'}],\phi_0)_{L^2(0,1)} & n=0;\\
\end{cases}
\]
for $n\in\mathbb{N}$, where $\phi_{n}$ was defined in
(\ref{eq:phi}). Then, equation (\ref{eq:eig:m}) becomes
\begin{align}
\label{eq:eig:m:2}
  \left[
  \begin{array}{ll}
     D_m - A_{mm} &  -\bR_m\\
   -\bC_m & \bI - \bB_m \\
  \end{array}
  \right]\left[
  \begin{array}{l}
    d_m\\
    \bc_m\\
  \end{array}
  \right] = {\bf 0},\quad m\in\mathbb{Z}.
\end{align}
We have the following lemma.
\begin{mylemma}
  \label{lem:inv:I-Bm}
  For $0<h\ll 1$ and $k\in{\cal B}$, $\bB_m$ is
  uniformly bounded from $\ell^2$ to $\ell^2$ with
  \begin{align}
    \label{eq:bounded:Bm}
    ||\bB_m + 2{\bP}_2|| = {\cal O}(h)\quad \mbox{as} \; h\to 0^+.
  \end{align}
  Moreover, $\bI-\bB_m$ and $\bI + 2\bP_2$ attain uniformly bounded inverses for $h\ll1$, and there holds
  \begin{align}
    \label{eq:inv:Bm}
    ||({\bI}-\bB_m)^{-1} - ({\bI}+2{\bP}_2)^{-1}|| = {\cal O}(h),\quad \quad \mbox{as} \; h\to 0^+.
  \end{align}
  In the above, the prefactors in the ${\cal O}$-notations depend only on $m$
  and ${\cal B}$. 
  \begin{proof}
    The estimation~\eqref{eq:bounded:Bm} follow from
    Lemma~\ref{lem:dm'm:h} (i)-(iv). Using the Neumann series, we see that (\ref{eq:inv:Bm}) holds if $\bI + 2\bP_2$ is invertible, which is true since ${\cal S}_0$
    is positive and bounded below \cite[Cor. 8.13]{mcl00}.
  \end{proof}
\end{mylemma}
From the invertbility of the operator $\bI-\bB_m$ in the above lemma, for each $m\in \mathbb{Z}$ the system (\ref{eq:eig:m:2}) can be further reduced to the following single nonlinear equation:
\begin{equation}
  \label{eq:single}
  \left[D_m(k) - A_{mm}(k)- \bR_m(k)(\bI-\bB_m(k))^{-1}\bC_m(k)  \right]d_m = 0,
\end{equation}
where we make the argument $k$ explicit to emphasize the dependence of the equation on $k$.
We call \eqref{eq:single} the characteristic equation, which is the resonance condition for the scattering problem \eqref{eq:problemE_1} - \eqref{eq:problemE_4}.

\subsubsection{Asymptotic analysis of resonances}
We are ready to state and prove our first main result for the resonances.
\begin{mytheorem}
  \label{thm:even:res}
  Assume that $h\ll 1$, the resonances for the scattering problem \eqref{eq:problemE_1} - \eqref{eq:problemE_4} in the bounded region $\mathcal{B}$ are given as follows:
  \begin{itemize}
      \item[(i)] For each given integer $m\neq0$, there exist a finite sequence of  resonance in $\mathcal{B}$
\begin{equation}
  \label{eq:even:res}
    k^*_{m,2m'}=k_{m,2m'} -\frac{m^2h}{2k_{m,2m'}}- \frac{2\Pi_m(k_{m,2m'},h)}{k_{m,2m'}l} + {\cal O}(h^2\log^2h), \quad  m'\in\mathbb{N}^* \; \mbox{is bounded},
\end{equation}
and a near-$|m|$ resonance
\begin{equation}
  \label{eq:even:nearm}
  k^*_m = |m|  - \frac{|m| h}{2} - \frac{|m| h}{l}\left[ (\tilde{\alpha}_m(|m|)-\alpha_m(|m|))
      +\bi(\tilde{\beta}_m(|m|)-\beta_m(|m|))\right] + {\cal O}(h^2\log h),
\end{equation}
where $k_{m,2m'} = \sqrt{m^2 + \frac{(2m')^2\pi^2}{l^2}}$, 
\begin{align}
   \label{eq:Pim}
    \Pi_m(k,h) =& \frac{(m^2-k^2)}{2\pi}h\log h  +2k^2h(\tilde{\alpha}_m(k) +\bi\tilde{\beta}_m(k)) - 2m^2h(\alpha_m(k)+\bi\beta_m(k))\nonumber\\
    &+(m^2-k^2)h(\bp^{T}(\bI+2\bP)^{-1}\bp),
\end{align}
and $\alpha_m$, $\beta_m$, $\tilde{\alpha}_m$ and
$\tilde{\beta}_m$ are defined in (\ref{eq:alpham:k})-(\ref{eq:tbetam:k}).
\item [(ii)] m = 0: there exist a finite sequence of resonance in $\mathcal{B}$
\begin{equation}
  \label{eq:even:res:m=0}
  k^*_{0,2m'} =k_{0,2m'}-2k_{0,2m'}\Pi_0(k_{0,2m'},h) + {\cal O}(h^2\log^2 h),
\end{equation}
where 
\begin{align}
  \label{eq:Pi0}
      \Pi_0(k,h) = & -h\frac{\log h}{4\pi} + \alpha_1(k)h + \bi
     \beta_1(k)h - h\bp^{T}(\bI+2\bP)^{-1}\bp.
\end{align}
  \end{itemize}
  Moreover, each resonance in \eqref{eq:even:res} - \eqref{eq:even:res:m=0} obtains an imaginary part of order ${\cal O}(h)$.
\begin{proof}
We obtain the resonances for the scattering problem   by solving for the characteristic values satisfying
\begin{equation}
    \label{eq:single:m}
    D_m(k) =A_{mm}(k)+\bR_m(k)(\bI-\bB_m(k))^{-1}\bC_m(k)
\end{equation}
for each $m \neq 0$, which reads
  \begin{equation}
    \label{eq:tmp2}
    s_{m0}^N(e^{\bi s_{m0}^N l/2}-e^{-\bi s_{m0}^N l/2}) = -\bi(e^{\bi s_{m0}^N l/2}+e^{-\bi s_{m0}^N l/2})\left[\Pi_m(k,h) + {\cal O}(h^2\log h)  \right].
  \end{equation}
Recall that $s_{m0}^N$ is defined by (\ref{eq-s}) with $\lambda_{m0}^N =(\beta_{m0}^N)^2$ given in (\ref{eq-lambdaN}). Note that $\beta_{m0}^N$ attains asymptotic expansion (\ref{eq:asy:betam2:N}) as $h \to 0$. We first find the resonances that are away from the integer number $m$ as $h \to 0$. More precisely, at such resonances, we have 
$$
\liminf_{h\to 0}|s_{m0}^N|>0.
$$
To proceed, note that $e^{\bi
    s_{m0}^N l/2}-e^{-\bi s_{m0}^N l/2}={\cal O}(h\log h)$ for $h\ll 1$, since
the r.h.s. of (\ref{eq:tmp2}) is ${\cal O}(h\log h)$. Therefore, we have for some $m'\in\mathbb{N}^*$ that
    \[
      \epsilon_{mm'}:=s_{m0}^Nl -2m'\pi =o(1),\quad {\rm as}\quad h\to 0^+,
    \]
    and \eqref{eq:tmp2} leads to
    \[
      1 - e^{\bi \epsilon_{mm'}} = \frac{2\bi \Pi_m(k,h)}{s^N_{m0}+\bi \Pi_m(k,h)} + {\cal O}(h^2\log h).
    \]
    By Taylor's expansion of $\log(1-2x/(s_{m0}^N+x))$ at $x=0$ and by $\Pi_m(k,h)={\cal
      O}(h\log h)$,
    \begin{align*}
      \epsilon_{mm'} = -\bi\log\left[1 - \frac{2 \bi \Pi_m(k,h)}{s_{m0}^N+\bi\Pi_m(k,h)} -{\cal O}(h^2\log h)\right] =\frac{-2\Pi_m(k,h)}{s_{m0}^N} + {\cal O}(h^2\log^2 h).
    \end{align*}
    Therefore, $\epsilon_{m,2m'}={\cal O}(h\log h)$ so that $s_{m0}^Nl=2m'\pi + {\cal
      O}(h\log h)$. But by (\ref{eq:asy:betam2:N}),
    \[
    k = \sqrt{(\lambda_{m0}^N)^2 + (s_{m0}^N)^2} = \sqrt{m^2- m^2 h + s_{m0}^2} + {\cal O}(h^2).
    \] We thus have
  \[
    k = k_{m,2m'} + {\cal O}(h\log h).
  \]
  Now, according to the definition of $\Pi_m$ in (\ref{eq:Pim}), we have
  \[
    \Pi_m(k,h) = \Pi_m(k_{m,2m'},h) + {\cal O}(h^2\log^2 h), 
  \]
  so that 
  \[
    \epsilon_{mm'} = -\frac{2l\Pi_m(k_{m,2m'},h)}{(2m')\pi} + {\cal O}(h^2\log^2 h).
  \]
  Hence
  \begin{align*}
    s^N_{m0} =& \frac{(2m')\pi}{l} -\frac{2\Pi_m(k_{m,2m'},h)}{(2m')\pi} + {\cal O}(h^2\log^2 h),
  \end{align*}
  and one obtains the expansion (\ref{eq:even:res}).
  Therefore, resonances $k$ satisfying
  (\ref{eq:single:m}) attain the asymptotic expansion (\ref{eq:odd:res}) for $h\ll 1$ for some $m'\in\mathbb{N}^*$.

  As for the existence of resonances, one notices that when $k$ lies in the region $D_h=\{k\in\mathbb{C}:{\rm Re}(k)>0,|s_{m0}^N(k)l -(2m')\pi|\leq
  h^{1/2}\}\subset {\cal B}$, the following holds on the boundary of this disk
   \begin{align*}
     &\Bigg|\left[(e^{\bi s_{m0}^N l}+1) -(\bi s_{m0}^N)^{-1}(e^{\bi s_{m0}^N l}-1)\left[ \Pi_{m}(k,h) +{\cal O}(h^2\log h)\right]  \right] - \left[\bi (s_{m0}^N l - (2m')\pi)  \right]\Bigg| \\
     &={\cal O}(h)\leq \sqrt{h}=|\bi (s_{m0}^N l - (2m')\pi)|.
   \end{align*}
   Rouch\'e's theorem states that there exists a unique root for
   (\ref{eq:single:m}) in $D_h$. Similarly one can verify the expansion (\ref{eq:even:res:m=0}) for $m=0$.
   
Finally, we solve for resonances that are asymptotically close to the integer $m$ when $h \to 0$. To do so, 
assume that $s_{m0}^N=o(1)$, as $h\to 0^+$. 
  Since $(e^{\bi s_{m0}^Nl /2}+e^{-\bi s_{m0}^N l/2}) = 1 + {\cal O}([s_{m0}^N]^2)$,
  we have
  \begin{align*}
    \bi [s_{m0}^N]^2l=& s_{m0}^N(e^{\bi s_{m0}^N l/2}-e^{-\bi s_{m0}^N l/2})+{\cal O}([s_{m0}^N]^4)\\
    =& -\bi\left[\Pi_m(k,h) + {\cal O}(h^2\log h)  \right] + {\cal O}([s_{m0}^N]^2h\log h) + {\cal O}([s_{m0}^N]^4)\\
    =& -\bi \Pi_m(m,h)+ {\cal O}(h^2\log h) + {\cal O}([s_{m0}^N]^2h\log h) + {\cal O}([s_{m0}^N]^4)\\
    =&-2m^2h\bi\left[ (\tilde{\alpha}_m(m) +\bi\tilde{\beta}_m(m)) - (\alpha_m(m)+\bi\beta_m(m)) \right]\\
    &+ {\cal O}(h^2\log h) + {\cal O}([s_{m0}^N]^2h\log h) + {\cal O}([s_{m0}^N]^4).
  \end{align*}
  Thus, $[s_{m0}^N]^2= {\cal O}(h)$ and
  \begin{align*}
[s_{m0}^N]^2=& -2m^2hl^{-1}\left[ (\tilde{\alpha}_m(m) +\bi\tilde{\beta}_m(m)) - (\alpha_m(m)+\bi\beta_m(m)) \right]+ {\cal O}(h^2\log h),
  \end{align*}
  which implies (\ref{eq:even:nearm}). 
\end{proof}
\end{mytheorem}
\begin{myremark}
   In fact, the constant $\bp^{T}(\bI+2\bP)^{-1}\bp$ in the expansions~(\ref{eq:even:res}) and~(\ref{eq:even:res:m=0}) is given explicitly as
  \[
    \bp^{T}(\bI+2\bP)^{-1}\bp = \frac{1}{2\pi^2} - \frac{1}{\pi^2}\log\left( \frac{\pi}{2} \right).
  \]
  It is obtained in \cite{holsch19} for the 2D slit problem using the matched asymptotic expansion and also numerically verified in \cite{zholu21}.
\end{myremark}
\begin{myremark}
The resonances attain the imaginary parts of order ${\cal O}(h)$, thus they are very close to the real axis when $h\ll1$. We point out that $k^*_{m,2m'}$ and $k^*_{m}$ given in \eqref{eq:even:res} and \eqref{eq:even:nearm} are resonances associated with the TE modes in the annular hole. Note that the leading-order of resonances $k^*_{m,2m'}$ depends on the metal thickness $l$, while the leading-order of resonances  $k^*_{m}$ is independent of $l$.  The independence on the metal thickness for the latter is also called epsilon-near-zero phenomenon \cite{yoo16}. On the other hand, $k^*_{0,2m'}$ given in \eqref{eq:even:res:m=0} are resonances associated with the TEM mode in the annular hole. As discussed in Section 5, the excitation of these two types of resonances are very different. 
\end{myremark}


\subsection{Resonances for Problem (O)}
In this section, we characterize the resonances for scattering
problem (O). Due to the similarity between the even problem \eqref{eq:problemE_1} - \eqref{eq:problemE_4} and the odd problem \eqref{eq:problemO_1} - \eqref{eq:problemO_3}, we shall
directly state the difference and the final results. By the same vectorial mode matching procedure, we can still obtain the linear system
(\ref{eq:INF:sys}) but with the following replacements: on the l.h.s,
\begin{align*}
  \sin(kl/2)\to \cos(kl/2),\quad \sin(s_{mn}^Nl/2)\to \cos(s_{mn}^Nl/2),\quad \sin(s_{ij}^Dl/2)\to \cos(s_{ij}^Dl/2);
\end{align*}
on the r.h.s.,
\begin{align*}
  \cos(kl/2)\to -\sin(kl/2),\quad \cos(s_{mn}^Nl/2)\to -\sin(s_{mn}^Nl/2),\quad \cos(s_{ij}^Dl/2)\to -\sin(s_{ij}^Dl/2),
\end{align*}
and the auxiliary coefficients $|2+\sin(s_{mn}^Nl/2)|$ and
$|2+\sin(s_{ij}^Dl/2)|$ on both sides keep unchanged. With the above minor
changes, we obtain the eigenvalue problem (\ref{eq:eig:m:2}) and the
characteristic equation (\ref{eq:single:m}) with
$D_m$, $A_{mm}$, $\bR_m$, $\bB_m$ and $\bC_m$ changed accordingly. 
From now on, we shall add the superscript $o$ (or $e$) to all the elements in  (\ref{eq:eig:m:2}) to indicate that they are for Problem (O) (or (E)).

The asymptotic analysis of the resonances are stated in the following theorem.
\begin{mytheorem}
  \label{thm:odd:res}
 Assume that $h\ll 1$, there exist a finite sequence of resonances for Problem
  (O) in $\mathcal{B}$ for each given integer $m\neq0$:
\begin{equation}
  \label{eq:odd:res}
    k^*_{m,2m'+1}=k_{m,2m'+1} -\frac{m^2h}{2k_{m,2m'+1}}- \frac{2\Pi_m(k_{m,2m'+1},h)}{k_{m,2m'+1}l} + {\cal O}(h^2\log^2h),\quad m'\in\mathbb{N},
\end{equation}
and a finite sequence of resonances when $m=0$:
\begin{equation}
  \label{eq:odd:res:m=0}
  k^*_{0,2m'+1} =k_{0,2m'+1}-2k_{0,2m'+1}\Pi_0(k_{0,2m'+1},h) + {\cal O}(h^2\log^2 h),\quad m'\in\mathbb{N},
\end{equation}
where $k_{m,2m'+1} = \sqrt{m^2 + \frac{(2m'+1)^2\pi^2}{l^2}}$.
\begin{proof}
 For the scattering problem (O), the characteristic equation (\ref{eq:single:m}) becomes
  \begin{align}
    \label{eq:tmp1}
    \bi s_{m0}^N(e^{\bi s_{m0}^N l}+1)
    =&(e^{\bi s_{m0}^N l}-1)\left[ \Pi_{m}(k,h) +{\cal O}(h^2\log h)\right]. 
  \end{align}
  In the following, we claim the trivial solution $k=\sqrt{\lambda_{m0}^N}$ is
  not a resonance. According to Lemmas~\ref{lem:dm'm:h} (v) and
  \ref{lem:inv:I-Bm}, $\bC^o_m\equiv 0$ so that $\bc^o_m\equiv 0$. But
  (\ref{eq:rep:tE}) and (\ref{eq:rep:tH}) imply $\nu\times \bE^o =\nu\times\bH^o =
  0$ on $A^h$ so that $\bE^o=\bH^o\equiv 0$ in the whole space
  $\mathbb{R}^3\backslash\overline{\Omega_{\rm M}}$. Moreover,
  \begin{align*}
    \lim_{h\to 0}\frac{2}{l} =& \lim_{h\to 0}\lim_{k\to \sqrt{\lambda_{m0}^N}}\frac{\bi s_{m0}^N(e^{\bi s_{m0}^N l}+1)}{(e^{\bi s_{m0}^N l}-1)} = \lim_{h\to 0}\lim_{k\to \sqrt{\lambda_{m0}^N}}\left[ \Pi_{m}(k,h) +{\cal O}(h^2\log h)\right] = 0,
  \end{align*}
  which is impossible. Thus, there is no resonance near $s_{m0}$ for
  problem (O), which is the main difference compared with problem (E). The proof of the expansions ~\eqref{eq:odd:res} and \eqref{eq:odd:res:m=0} for the resonances follow the same lines as Theorem \ref{thm:even:res}.
 \end{proof}
\end{mytheorem}

\section{Electromagnetic field enhancement at resonant frequencies}

In this section, we solve the scattering problem (\ref{eq:E}) - (\ref{eq:smc}) when the incident wave $\{ {\bf E}^{\rm inc}, {\bf H}^{\rm inc} \}$ is present and study the electromagnetic field enhancement.

\subsection{Field enhancement due to the excitation of a TE mode in the annular hole}
Let us first consider the scattering problem when the incident frequency coincides with the real part of the resonance $k_{1,m'}^*$ ($m'\in\mathbb{N}^*$) in \eqref{eq:even:res} or \eqref{eq:odd:res}, or the resonance $k_1^*$ in \eqref{eq:even:nearm}.
For conciseness of the presentation, we only show the calculations for the normal incidence such that the polarization vectors in \eqref{eq:pla:inc} are given by  ${\bf E}^0=(0,1,0)^{T}$ and ${\bf H}^0=(1,0,0)^{T}$.

\begin{mytheorem}
\label{thm:enhance:te}
For a normal incident wave with the polarization vectors ${\bf E}^0=(0,1,0)^{T}$ and ${\bf H}^0=(1,0,0)^{T}$, the magnitude of electromagnetic field $\bE$ and $\bH$ in the hole $G^h$ attains the order ${\cal O}(h^{-1})$ at resonant frequencies ${\rm Re}(k_{1,m'}^*)$ for each $m'\in\mathbb{N}^*$ or ${\rm Re}(k_1^*)$. Specifically,
\begin{itemize}
  \item[(1).] If $k={\rm Re}(k_{1,m'}^*)$ for an even integer $m'$, 
  the following expansions hold inside the hole $G^{h}_+$:
    \begin{align}
    \label{eq:H3e:case1}
     H_3(x) =&
     h^{-1}(-1)^{m'/2}\frac{\cos(m'\pi/l x_3)x_1 e^{\bi k_{1,m'}l/2}\bi}{2[k_{1,m'}^{2}\tilde{\beta}_1(k_{1,m'})-\beta_1(k_{1,m'})]} + {\cal O}(\log h),\\
  E_1(x) 
  =&h^{-1}(-1)^{m'/2}\frac{k_{1,m'}e^{\bi k_{1,m'}l/2}\cos(m'\pi/l x_3)x_1x_2  }{2[k_{1,m'}^{2}\tilde{\beta}_1(k_{1,m'})-\beta_1(k_{1,m'})]} + {\cal O}(\log h)\\
  E_2(x) 
  =&-h^{-1}(-1)^{m'/2}\frac{k_{1,m'}e^{\bi k_{1,m'}l/2}\cos(m'\pi/l x_3)x_2^2  }{2[k_{1,m'}^{2}\tilde{\beta}_1(k_{1,m'})-\beta_1(k_{1,m'})]} + {\cal O}(\log h).
\end{align} 
  \item[(2).] If $k={\rm Re}(k_{1,m'}^*)$ for an odd integer $m'$, then in the hole $G^{h}_+$, we have
    \begin{align}
    \label{eq:H3o:case1}
     H_3(x) =&
      h^{-1}(-1)^{(m'-1)/2}\frac{\sin(m'\pi/l x_3)x_1e^{\bi k_{1,m'}l/2}\bi}{2[k_{1,m'}^{2}\tilde{\beta}_1(k_{1,m'})-\beta_1(k_{1,m'})]} + {\cal O}(\log h),\\
  E_1(x) 
  =&h^{-1}(-1)^{(m'-1)/2}\frac{k_{1,m'}e^{\bi k_{1,m'}l/2}\sin(m'\pi/l x_3)x_1x_2  }{2[k_{1,m'}^{2}\tilde{\beta}_1(k_{1,m'})-\beta_1(k_{1,m'})]} + {\cal O}(\log h),\\
  E_2(x) 
  =&-h^{-1}(-1)^{(m'-1)/2}\frac{k_{1,m'}e^{\bi k_{1,m'}l/2}\sin(m'\pi/l x_3)x_2^2  }{2[k_{1,m'}^{2}\tilde{\beta}_1(k_{1,m'})-\beta_1(k_{1,m'})]} + {\cal O}(\log h).
\end{align} 
\item[(3).] If $k={\rm Re}(k_1^*)$, then 
      \begin{align}
      \label{eq:H3e:case2}
      H_3(x)
        =&\frac{2\bi h^{-1}e^{\bi l/2}x_1}{\tilde{\beta}_1(1)-\beta_1(1)}+{\cal O}(\log h),\\
       E_1(x) =& \frac{2h^{-1}e^{\bi l/2}x_1x_2}{\tilde{\beta}_1(1)-\beta_1(1)}+{\cal O}(\log h),\\
  E_2(x) =& -\frac{2h^{-1}e^{\bi l/2} x_2^2}{\tilde{\beta}_1(1)-\beta_1(1)}+{\cal O}(\log h).
\end{align}
\end{itemize}  
\end{mytheorem}
\begin{proof}
 For the normal incidence, the reflected field is
\[
  \bE^{\rm ref} = -\bE^0 e^{\bi k(x_3-l)}, \bH^{\rm ref} = \bH^0 e^{\bi k(x_3-l)} \; \mbox{for} \; x_3>l/2.
\]  
The total field can be decomposed as
\begin{align*}
  \bE(x) = \left\{
  \begin{array}{lc}
    \bE^{\rm e}(x) + \bE^{\rm o}(x), & x_3\geq 0,\\
     \bE^{\rm e,*}(x^*) - \bE^{\rm o,*}(x^*), & x_3<0,\\
  \end{array}
  \right.
  \quad \bH(x) = (\bi k)^{-1}\cl\ \bE,
\end{align*}
wherein $\{\bE^{\rm e}, \bH^{\rm e}\}$ and $\{\bE^{\rm o}, \bH^{\rm o}\}$ satisfy \eqref{eq:problemE_1}-\eqref{eq:problemE_4} and \eqref{eq:problemO_1}-\eqref{eq:problemO_3} respectively.
On the annular aperture $A^h$, using the integral equation (\ref{eq:t2t:0}), 
it follows that
\begin{equation*}
  \nu\times \bH^{j}|_{A^h} +\frac{2}{\bi k}{\cal L}_k[\nu\times\bE^{j}|_{A^h}] = \nu\times \frac{\bH^{\rm inc}+\bH^{\rm ref}}{2}|_{A^h} = \left[\begin{matrix}
    0\\
    -e^{\bi kl/2} \\
    0\\
    \end{matrix}\right]= \left[\begin{matrix}
    -e^{\bi kl/2}\nabla_2(r\sin\theta)\\
    0\\
    \end{matrix}\right],\quad j=e,o.
\end{equation*}
Then, using $\sin\theta=
(2\bi)^{-1}(e^{\bi \theta}-e^{-\bi\theta})$ and Lemma~\ref{lem:orth:theta}, the mode-matching procedure in Section
4.1 gives rise to only two inhomogeneous INF linear system as shown below:
\begin{align}
  \label{eq:inh:sys}
\left[
  \begin{array}{ll}
     D_m^{j} - A^{j}_{mm} & -\bR^{j}_{m}\\
    -\bC^{j}_m & \bI - \bB^{j}_{m}\\ 
  \end{array}
  \right]\left[
  \begin{array}{l}
    d_m^{j}\\
    \bc_m^{j}\\
  \end{array}
  \right] = c^{j}\left[
  \begin{array}{l}
    a_m\\
    \bb_m
  \end{array}
  \right],\quad j=e,o;m=\pm 1.
\end{align}
In the above, $c^e=1$, $c^o=\bi$, $a_m = \frac{ke^{\bi
    kl/2}}{2\bi\lambda_{m0}^N}\langle \nabla_2 r\sin\theta, {\cal
  R}\cl_2\overline{\psi_{m0}^N} \rangle$, and
\begin{align}
  \bb_m &= \left[
  \begin{matrix}
   \{\frac{ke^{\bi kl/2}}{2\bi s_{mn'}^N(\lambda_{mn'}^N)^{1/4}}\langle \nabla_2 r\sin\theta, {\cal R}\cl_2\overline{\psi_{mn'}^N} \rangle\}_{n'\in\mathbb{N}^*}\\ 
   \{\frac{e^{\bi kl/2}s_{mj'}^D}{2\bi k(\lambda_{mj'}^D)^{3/4}}\langle \nabla_2 r\sin\theta, {\cal R}\nabla_2\overline{\psi_{mj'}^D} \rangle\}_{j'\in\mathbb{N}^*}\\ 
  \end{matrix}
  \right]\nonumber\\
  &=\left[
  \begin{matrix}
   \{\frac{ke^{\bi kl/2}(\lambda_{mn'}^N)^{3/4}}{2\bi s_{mn'}^N}\langle r\sin\theta, \overline{\psi_{mn'}^N} \rangle\}_{n'\in\mathbb{N}^*}\\ 
   \{0 \}_{j'\in\mathbb{N}^*}\\ 
  \end{matrix}
  \right],
\end{align}
where the last equality holds due to Green's identities.
We study the enhancement of the electromagnetic field $\{\bE^{e},\bH^{e}\}$ in the hole $G^{h}_+$ first. 

By Lemma~\ref{lem:psimnN}, integrating by parts gives
\[
  a_{m} = \frac{ke^{\bi kl/2}}{2\bi}\langle r\cos\theta, \overline{\psi_{m0}^N}
  \rangle_{A^h} = \frac{ke^{\bi kl/2}}{2\bi}\left[\sqrt{\frac{\pi h}{2}}+{\cal O}(h^{3/2})  \right]
\]
and 
\[
\langle  r\cos\theta, \overline{\psi_{mn'}^N} \rangle_{A^h} = {\cal O}[(n')^{-2}h^{3/2}],
\]
so that
$||\bb_{m}||_{\ell^2} = {\cal O}(h)$ . 
Using Lemmas~\ref{lem:inv:I-Bm} and \ref{lem:dm'm:h} (vi), the system (\ref{eq:inh:sys}) can be
reduced to the following inhomogeneous equation:
\begin{align*}
  \left[D^e_m(k) - A^e_{mm}(k)- \bR^e_m(k)(\bI-\bB^e_m(k))^{-1}\bC^e_m(k)  \right]d^e_m =& a_m + \bR^e_m(\bI-\bB^e_m)^{-1}\bb_m  \\
  =& \frac{ke^{\bi kl/2}\sqrt{2\pi h}}{4\bi}+{\cal O}(h^{3/2}), 
\end{align*}
for $m=\pm 1$.

\bigskip

\noindent {\it (1)}. Let $k={\rm Re}(k_{1,m'}^*)$ for  an even integer $m'\in\mathbb{N}^*$, in which $k_{1,m'}^*$ is given by \eqref{eq:even:res}. Since $k - k_{1,m'}^*=-{\rm Im}(k_{1,m'}^*)\bi ={\cal O}(h)$, we obtain
    \begin{align*}
      D^e_m(k) - D^e_m(k_{1,m'}^*) =& -[D_m^e]'(k_{1,m'}){\rm Im}(k_{1,m'}^*)\bi + {\cal O}(h^2\log h)\\
      =&2h(-1)^{m'/2}\left[k_{1,m'}^2\tilde{\beta}_1(k_{1,m'})-\beta_1(k_{1,m'})\right]\bi + {\cal O}(h^2\log h),\\
      A^e_{mm}(k) - A^e_{mm}(k_{m,m'}^*)=&{\cal O}(h^2\log h),\\
      \bR^e_m(k)(\bI-\bB^e_m(k))^{-1}\bC^e_m(k)=&\bR^e_m(k_{1,m'}^*)(\bI-\bB^e_m(k_{1,m'}^*))^{-1}\bC^e_m(k_{1,m'}^*)+{\cal O}(h^2\log h).
    \end{align*}
    Therefore,
    \begin{align*}
      d^e_m =& -\frac{k_{1,m'}(-1)^{m'}e^{\bi k_{1,m'}l/2}\sqrt{2\pi h}+{\cal O}(h^{3/2})}{8 h[k_{1,m'}^{2}\tilde{\beta}_1(k_{1,m'})-\beta_1(k_{1,m'})]+{\cal O}(h^2\log h)}\\
      =& -\frac{h^{-1/2}k_{1,m'}(-1)^{m'/2}e^{\bi k_{1,m'}l/2}\sqrt{2\pi }}{8 [k_{1,m'}^{2}\tilde{\beta}_1(k_{1,m'})-\beta_1(k_{1,m'})]}(1+{\cal O}(h\log h))
    \end{align*}
    and
    \begin{align*}
     ||\bc^e_m||_{\ell^2}=& ||(\bI-\bB^e_m)^{-1}(\bb_m+\bC^e_md^e_m)||_{\ell^2} = {\cal O}(1),
    \end{align*}
    for $m=\pm 1$. Hence, using the field representation (\ref{eq:rep:bH}), inside the hole
    $G^{h}_+$, there holds
    \begin{align*}
     H_3^e(x) =& \sum_{m\in\{-1,1\}}\sum_{n=1}^{\infty}[d_{mn}^{TE}]^e\frac{-2\bi\lambda_{mn}^N\cos(s_{mn}^Nx_3)}{k} \psi_{mn}^N(r,\theta)\nonumber\\
      &+\sum_{m\in\{-1,1\}}[d_{m0}^{TE}]^e\frac{-2\bi\lambda_{m0}^N\cos(s_{m0}^Nx_3)}{k} \psi_{m0}^N(r,\theta)\nonumber\\
      =&\sum_{m\in\{-1,1\}}\sum_{n=1}^{\infty}[c_{m;n}^{TE}]^e\frac{-2\bi[\lambda_{mn}^N]^{1/4}\cos(s_{mn}^Nx_3)}{k\sin(s_{mn}^Nl/2)} \psi_{mn}^N(r,\theta)\nonumber\\
      &+\sum_{m\in\{-1,1\}}d_{m}^e\frac{-2\bi\lambda_{m0}^N\cos(s_{m0}^Nx_3)}{k} \psi_{m0}^N(r,\theta)\nonumber\\
      =&h^{-1}(-1)^{m'/2}\frac{\cos(m'\pi/l x_3)\cos(\theta)e^{\bi k_{1,m'}l/2}\bi}{2[k_{1,m'}^{2}\tilde{\beta}_1(k_{1,m'})-\beta_1(k_{1,m'})]} + {\cal O}(\log h).
      \end{align*}
Similarly, an application of 
(\ref{eq:rep:bE}) yields the asymptotic for $E_1^e(x)$ and $E_2^e(x)$. 

\bigskip

\noindent{\it (2).} The case when $m'$ is odd can be derived similarly as above.  
   
\bigskip

\noindent{\it (3).} Let $k={\rm Re}(k_1^*)$. Again, $k-k_1^*=-{\rm Im}(k_1^*)\bi={\cal O}(h)$ so that 
      \begin{align*}
       D^e_m(k)-D^e_m(k_1^*) =& -[D^e_m]'(k_1^*){\rm Im}(k_1^*)\bi + {\cal O}(h^2)\\ 
        =&\frac{h}{2}(\tilde{\beta}_1(1)-\beta_1(1))\bi + {\cal O}(h^2),\\
      A^e_{mm}(k) - A^e_{mm}(k_{1}^*)=&{\cal O}(h^2\log h),\\
      \bR^e_m(k)(\bI-\bB^e_m(k))^{-1}\bC^e_m(k)=&\bR^e_m(k_{1}^*)(\bI-\bB^e_m(k_{1}^*))^{-1}\bC^e_m(k_{1}^*)+{\cal O}(h^2\log h),
      \end{align*}
      Thus, 
      \begin{align*}
        d^e_m 
        = \frac{-h^{-1/2}e^{\bi l/2}\sqrt{2\pi }}{2((\tilde{\beta}_1(1)-\beta_1(1))}(1 + {\cal O}(h\log h)),
      \end{align*}
      and $||\bc^e_m||_{\ell^2}={\cal O}(1)$. Inside the hole $G^{h}_+$, 
      \begin{align*}
      \label{eq:H3e:case2}
      H_3^e(x)=&\sum_{m\in\{-1,1\}}\sum_{n=1}^{\infty}[c_{m;n}^{TE}]^e\frac{-2\bi[\lambda_{mn}^N]^{1/4}\cos(s_{mn}^Nx_3)}{k\sin(s_{mn}^Nl/2)} \psi_{mn}^N(r,\theta)]\nonumber\\
      &+\sum_{m\in\{-1,1\}}d_{m}^e\frac{-2\bi\lambda_{m0}^N\cos(s_{m0}^Nx_3)}{k} \psi_{m0}^N(r,\theta)\nonumber\\
        =&\frac{2\bi h^{-1}e^{\bi l/2}\cos\theta}{\tilde{\beta}_1(1)-\beta_1(1)}+{\cal O}(\log h),
      \end{align*}
and similarly, the asymptotic expansions for  $E_1^e(x)$ and $E_2^e(x)$ can be derived. 
\end{proof}



From \eqref{eq:inh:sys}, we see that
a normal incident plane wave can only excite the TE$_{mn}$ modes in the annular hole with $m=\pm1$. To excite higher-order modes with $|m|\ge2$, an oblique incident wave needs to be applied.
Without loss of generality, we assume that the incident direction $d=(d_1,0,-d_3)^{T}$ and the electric polarization vector
$\bE_0=(0,1,0)^{T}$. Then repeating the above procedure gives
\begin{equation}\label{eq:obliq_src}
  \nu\times \frac{\bH^{\rm inc}+\bH^{\rm ref}}{2}|_{A^h} =\left[
    \begin{matrix}
      0\\
-d_3e^{-\bi k d_3l/2+\bi kd_1x_1}\\
0\\
\end{matrix}
  \right]=\left[  
    \begin{matrix}
d_3e^{-\bi k d_3l/2}\cl_2\frac{e^{\bi kd_1 r\cos\theta}-1}{\bi kd_1}\\
0\\
\end{matrix}
  \right].
\end{equation}
From the Jacobi-Anger
expansion \cite[Eq. (3.89)]{colkre13},
\begin{align}
  \frac{e^{\bi kd_1 r\cos\theta}-1}{\bi k d_1}=(\bi k d_1)^{-1}\sum_{m=-\infty}^{\infty}[\bi^m J_m(k d_1 r)-\delta_{m0}]e^{\bi m\theta},
\end{align}
the expansion contains terms with higher-order angular momentum. Therefore, the enhancement of  the electromagnetic field $\{\bE,\bH\}$ can be obtained at the resonant frequencies $k={\rm
  Re}(k_{m,m'}^*)$ with $|m|\ge 2$. We omit the detailed calculations here.


\subsection{Field enhancement due to excitation of the TEM mode in the annular hole}
In this section, we consider field enhancement at the resonant frequencies $k={\rm Re}(k_{0,2m'}^*)$ in \eqref{eq:even:res:m=0}, which are associated with the TEM mode in the annular hole. When a plane wave impinges on the subwavelength structure, the source term in \eqref{eq:obliq_src} is orthogonal to the resonant mode in the sense that
\[
 \langle\cl_2\frac{e^{\bi kd_1 r\cos\theta}-1}{\bi kd_1}, {\cal R}\nabla_2\log r
 \rangle = -\langle\cl_2\frac{e^{\bi kd_1 r\cos\theta}-1}{\bi kd_1}, \nabla_2\log r
 \rangle_{A^h} = 0.
\]
This follows from the orthogonality relation $\cl_2 H^1(A^h) \perp \mathbb{H}_2(A^h)$ in Lemma~\ref{lem:helmdecomp}.
Thus the mode matching formulation leads to a homogeneous system for $m=0$, which only attains trivial solution for $k\in\mathbb{R}$. This implies that no field amplification could be obtained at the resonant frequencies 
$k={\rm Re}(k_{0,m'}^*)$. In other words, TEM modes can not be excited by using the plane wave incidence.

To excite a TEM mode, we consider a spherical incident wave produced by an electric monopole located at $(0,0, y_3)^{T}$ with $y_3>0$ and it points toward the $x_3$-direction. Namely, $\bE^{\rm inc}$ satisfies
\[
\cl\ \cl\ \bE^{\rm inc} - k^2 \bE^{\rm inc}= -\hat{x}_3\delta(0,0,x_3-y_3),
\]
where $\hat{x}_3 = (0,0,1)^{T}$ is the unit vector in $x_3$-direction. 
Indeed, it is known that 
\begin{equation}
\label{eq:sph:inc}
\bE^{\rm inc} = \left[\hat{x}_3 + k^{-2}\partial_{x_3}\nabla\right]\left(\frac{e^{\bi k\sqrt{r^2+(x_3-y_3)^2}}}{4\pi \sqrt{r^2+(x_3-y_3)^2}}\right).
\end{equation}
The reflected electric field produced by the perfect conducting metallic slab is given by
\[
\bE^{\rm ref} = \left[\hat{x}_3 + k^{-2}\partial_{x_3}\nabla\right]\left(\frac{e^{\bi k\sqrt{r^2+(x_3+y_3)^2}}}{4\pi \sqrt{r^2+(x_3+y_3)^2}}\right).
\]
We have the following result for the field amplification. 
\begin{mytheorem} \label{thm-enhance-tem}
Let the incident electric field be of the form (\ref{eq:sph:inc}). The magnitude of the total electromagnetic field $\bE$ and $\bH$ attains the order ${\cal O}(h^{-1})$ at frequencies $k={\rm Re}(k_{0,m'}^*)$ for $m'\in\mathbb{N}^*$. Specifically,
\begin{itemize}
    \item[(1).] If
$k={\rm Re}(k_{0,m'}^*)$ for an even integer $m'$, the following hold in the annular gap $G_+^h$ :
  \begin{align}
       E_1 =& \frac{2(-1)^{m'/2}c(k_{0,m'}, y_3)}{h} \, x_1\, \cos(k_{0,m'}x_3) + {\cal O}(\log h),\\ 
       E_2 =& \frac{2(-1)^{m'/2}c(k_{0,m'}, y_3)}{h} \, x_2 \, \cos(k_{0,m'}x_3) + {\cal O}(\log h),\\ 
       H_1 =& -\frac{2(-1)^{m'/2}\bi c(k_{0,m'}, y_3)}{h}  \, x_2 \,  \sin( k_{0,m'} x_3) + {\cal O}(\log h),\\ 
       H_2 =& \frac{2(-1)^{m'/2}\bi c(k_{0,m'}, y_3)}{h} \, x_1 \, \sin(k_{0,m'} x_3) + {\cal O}(\log h).
    \end{align}
    where $c(k,y) = -\frac{e^{\bi k\sqrt{1+y^2}}(\bi k\sqrt{1+y^2}-1)}{2m'\pi\beta_1(k)(1+y^2)^{3/2}}$.
    \item[(2).] If $k={\rm Re}(k_{0,m'}^*)$ for an odd integer $m'$, then in the annular gap $G_+^h$,
  \begin{align}
       E_1 =& \frac{2(-1)^{(m'-1)/2}c(k_{0,m'}, y_3)}{h} \, x_1\, \sin(k_{0,m'}x_3) + {\cal O}(\log h),\\ 
       E_2 =& \frac{2(-1)^{(m'-1)/2}c(k_{0,m'}, y_3)}{h} \, x_2 \, \sin(k_{0,m'}x_3) + {\cal O}(\log h),\\ 
       H_1 =& -\frac{2(-1)^{(m'-1)/2}\bi c(k_{0,m'}, y_3)}{h}  \, x_2 \,  \cos( k_{0,m'} x_3) + {\cal O}(\log h),\\ 
       H_2 =& \frac{2(-1)^{(m'-1)/2}\bi c(k_{0,m'}, y_3)}{h} \, x_1 \, \cos(k_{0,m'} x_3) + {\cal O}(\log h).
   \end{align}
\end{itemize}
\end{mytheorem}
\begin{proof}
We have on the annular aperture $A^h$, 
\[
  \nu\times \frac{\bH^{\rm inc}+\bH^{\rm ref}}{2}|_{A^h} =F(r,y_3)\left[
  \begin{matrix}
  \nabla_2\log r\\
  0
  \end{matrix}
  \right],
\]
where
\[
F(r,y_3) = \frac{(k\sqrt{r^2+y_3^2}+\bi)r}{k(r^2+y_3^2)^{3/2}}e^{\bi k \sqrt{r^2+y_3^2}}.
\]
Following our mode matching procedure, we obtain the systems below analogous to (\ref{eq:inh:sys}):
\begin{align}
\label{eq:sys:TEM:m}
\left[
  \begin{array}{ll}
     D_m^{j} - A^{j}_{mm} & -\bR^{j}_{m}\\
    -\bC^{j}_m & \bI - \bB^{j}_{m}\\ 
  \end{array}
  \right]\left[
  \begin{array}{l}
    d_m^{j}\\
    \bc_m^{j}\\
  \end{array}
  \right] = c^{j}\left[
  \begin{array}{l}
    a_m\\
    \bb_m
  \end{array}
  \right],  \quad j=e,o, \quad m\in\mathbb{Z}.
\end{align}
In the above, the source term
\begin{equation} \label{eq-am}
a_m = \begin{cases}
  \frac{k\bi }{4\pi \log(1+h)}\langle F(r,y_3)\nabla_2\log r, {\cal R} \nabla_2 \log r \rangle, & m=0;\\
  \frac{k\bi }{2\lambda_{m0}^N}\langle F(r,y_3)\nabla_2\log r, {\cal R} \cl\overline{\psi_{m0}^N} \rangle, & m\neq 0,
\end{cases}
\end{equation}
and 
\[
\bb_m = \left[
  \begin{matrix}
   \{\frac{k\bi }{2s_{mn'}^N(\lambda_{mn'}^N)^{1/4}}\langle F(r,y_3)\nabla_2 \log r, {\cal R}\cl_2\overline{\psi_{mn'}^N} \rangle \}_{n'\in\mathbb{N}^*}\\ 
   \{\frac{\bi s_{mj'}^D}{2k(\lambda_{mj'}^D)^{3/4}}\langle F(r,y_3)\nabla_2 \log r, {\cal R}\nabla_2\overline{\psi_{mj'}^D} \rangle\}_{j'\in\mathbb{N}^*}\\ 
  \end{matrix}
  \right].
\]
We can derive from Lemma~\ref{lem:orth:theta} that 
$a_m=0$ and $\bb_m=0$ for $m\neq 0$. As such 
we need only focus on the case $m=0$. We further restrict to the case $j=e$ as the case $j=o$ can be dealt with in a similar manner. 

A direct calculation shows that
\begin{align*}
    a_0 
&=-\frac{\bi }{2\log(1+h)}\int_{1}^{1+h} \frac{(k\sqrt{r^2+y_3^2}+\bi)}{(r^2+y_3^2)^{3/2}}e^{\bi k \sqrt{r^2+y_3^2}}dr \\
& = -\frac{\bi }{2} \frac{(k\sqrt{1+y_3^2}+\bi)}{(1+y_3^2)^{3/2}}e^{\bi k \sqrt{1+y_3^2}} + {\cal O}(kh), \\
||\bb_0||_{\ell^2}&= {\cal O}(h),\\
 \bR^e_0(\bI-\bB^e_0)^{-1}\bb_0 &= {\cal O}(h).
\end{align*}

At $k={\rm Re}(k_{0,m'}^*)$ for an even integer $m'\in\mathbb{N}^*$, noting that $k-k_{0,2m'}^*=-{\rm Im}(k_{0,2m'}^*)\bi={\cal O}(h)$, we have 
    \begin{align*}
        D_0^e(k)-D_0^e(k_{0,m'}^*) =& -l/2\cos(k_{0,m'}^*l/2) {\rm Im}(k_{0,m'}^*)\bi + {\cal O}(h^2)\\
        =&(-1)^{m'/2}m'\pi\beta_1(k_{0,m'})h\bi  + {\cal O}(h^2\log h),\\
        A^e_{00}(k)-A^e_{00}(k_{0,m'}^*)=&{\cal O}(h^2\log h),\\
        \bR^e_0(k)(\bI-\bB^e_0(k))^{-1}\bC^e_0(k)=&\bR^e_0(k_{0,m'}^*)(\bI-\bB^e_0(k_{0,m'}^*))^{-1}\bC^e_0(k_{0,m'}^*)+{\cal O}(h^2\log h).
    \end{align*}
    Hence, from the equation 
    $$\left[D^e_0(k) - A^e_{00}(k)- \bR^e_0(k)(\bI-\bB^e_0(k))^{-1}\bC^e_0(k)  \right]d^e_0 = \, a_0 + \bR^e_0(\bI-\bB^e_0)^{-1}\bb_0,
    $$
 we can derive that
    \begin{align*}
    d_0^e =& -(-1)^{m'/2}h^{-1}\frac{e^{\bi k_{0,m'}\sqrt{1+y_3^2}}( k_{0,m'}\sqrt{1+y_3^2}+\bi)}{2m'\pi\beta_1(k_{0,m'})(1+y_3^2)^{3/2}}(1+{\cal O}(h\log h)),
    \end{align*}
    and by Lemma~\ref{lem:dm'm:h} (v), there holds
    \begin{align*}
    ||c_0^e||_{\ell^2}=||(\bI-\bB_0^e)^{-1}(\bb_0 + \bC_0^ed_0^e)||_{\ell^2} = {\cal O}(1).
    \end{align*}
    Therefore, using the field representations in (\ref{eq:rep:bE}) and (\ref{eq:rep:bH}), we obtain the desired asymptotic for the electromagnetic fields in the annular gap $G^h$.
\end{proof}
\begin{myremark}
We note that the electromagnetic field in the annular hole is amplified with the order ${\cal O}(h^{-1})$ at the resonant frequency $k={\rm Re}(k_{0,m'}^*)$ for some $m'\in\mathbb{N}^*$. Moreover, the waves oscillates along the $x_3$ direction but it varies linearly in the narrow annular cross section.
\end{myremark}

\section{Discussion and conclusion}

In this section, we discuss how the problem geometry and the topology of the subwavelength hole may affect the resonances and field enhancement for the scattering problem \eqref{eq:E}-\eqref{eq:smc}.

First, as pointed out at the beginning, for clarity the analysis is only presented for the inner radius of the annulus $a=1$. If $a\neq 1$, by the change of the scale, the roots $k$ of the characteristic equation (\ref{eq:single:m}) and the thickness $l$ are replaced by $ka$ and $l/a$, respectively. Thus the value of $a$ could significantly affect the resonant frequencies given in (\ref{eq:even:res}), (\ref{eq:even:nearm}), (\ref{eq:even:res:m=0}), (\ref{eq:odd:res}), and (\ref{eq:odd:res:m=0}). In practice, one can tune this parameter for applications in different frequency regimes. 
 Moreover, we note that $h$  is the relative width of the gap $G^h$. Thus one can increase the inner radius $a$ while keeping the absolute gap width $d=ah$ invariant. This will further increase the electromagnetic enhancement to the order ${\cal O}(h^{-1}) = d^{-1}{\cal O}(a)$.  

Note that we have assumed that the metal thickness $l={\cal O}(1)$ throughout the paper, and the prefactors in the error terms of the resonance formulae 
(\ref{eq:even:res}), (\ref{eq:even:nearm}), (\ref{eq:even:res:m=0}), (\ref{eq:odd:res}), and (\ref{eq:odd:res:m=0}) depends on $l$.  One natural question is how large $l$ is allowed so that our analysis still holds true. Compare the expressions in Theorems~\ref{thm:enhance:te} and \ref{thm-enhance-tem}. The leading terms of the fields due to TE modes do not change significantly as $l$ decreases or increases; however, the leading terms of the fields due to the TEM mode change significantly in terms of order $l^3{\cal O}(h^{-1})$ as $l$ increases, since $\beta_1(k_{0,m'}^*) = {\cal O}(l^{-3})$ as $l\to\infty$.  
We can carry out more delicate analysis to quantify the dependence of the resonances on $l$ more precisely and for $l$ not necessarily small. In fact, to ensure the existence of finite resonances for $m\neq 0$, it is sufficient to assume that $lh\log h\ll 1$. Let us revisit the characteristic equation (\ref{eq:tmp2}):
  \begin{equation*}
    s_{m0}^N(e^{\bi s_{m0}^N l/2}-e^{-\bi s_{m0}^N l/2}) = -\bi(e^{\bi s_{m0}^N l/2}+e^{-\bi s_{m0}^N l/2})\left[\Pi_m(k,h) + {\cal O}(h^2\log h)  \right].
  \end{equation*}
If $lh\log h\ll 1$, there holds 
 \[
 s_{m0}^N l \tan(s_{m0}^N l/2) = -l[\Pi_m(k,h)+{\cal O}(h^2\log h)] = {\cal O}(hl\log h)\ll 1.
\]
Thus $s_{m0}^N l\ll 1$ or $s_{m0}^N l/2 - m'\pi\ll 1$ for $m'\in\mathbb{N}^*$. Then following the parallel lines as the proof of Theorem~\ref{thm:even:res}, it can be shown that 
\[
{\rm Im}(k_{m,m'}^*) = {\cal O}(l^{-1}h\log h),\quad m'\in\mathbb{N}^*,
\]
where the prefactors no longer depend on $l$. 
However, the configuration when $l=\infty$ is more subtle, since the problem is not posed in an open medium anymore. One needs to define the corresponding scattering problem  properly and impose the radiation conditions carefully. The other extreme case is when the metal is infinitely thin with $l=0$. In such scenario, the hole no longer supports waveguide modes and a totally different approach needs to be adopted for analyzing the scattering problem. This is beyond the scope of this paper and will be investigated in a separate work.

We would also like to point out there are no resonances in the region ${\cal B}$ when the narrow annular hole is replaced by a tiny hole with a simply connected cross section, such as a tiny hollow hole with circular cross section.  Assume that the radius of the circle is given by $h\ll 1$. We can still apply the multiscale analysis framework in this paper for analyzing the resonances. However, the two eigenvalue problems (DEP) and (NEP) attain the eigenvalues of order ${\cal O}(1/h)$, though the corresponding eigenfunctions can still be used to construct the two function spaces $\cl_2 H^1(A^h)$ and $\nabla_2 H_0^1(A^h)$. In addition, the finite-dimensional space $\mathbb{H}_2(A^h)$ becomes $\{{\bf 0}\}$. Therefore, the mode matching procedure leads to a system analogous to (\ref{eq:eig:m:2}), with its first row and first column removed. This new system possesses only the zero solution for $k\in{\cal B}$ by Lemma~\ref{lem:inv:I-Bm}. In other words, there are no resonances in ${\cal B}$. 

Finally, we point out several directions along the line of this research. In this work, we focus on the resonances induced by the subwavelength annulus gap. Another type of resonance is related to surface plasmon, and the quantification of its effect on the overall resonant behavior of the structure is still open \cite{garmarebbkui10}. There are some preliminary studies of the plasmonic effect on the resonances in \cite{haftel2006} for the annular hole, but the understanding of the interactions between two types of resonances is far from complete. Another direction is to investigate the resonant scattering in more sophisticated structures, such as an array of annular holes, or the bull's eye structure, etc \cite{garmarebbkui10, yoo16}. The resonant phenomena become richer in those structures. In terms of applications, there are also several topics that need to be explored. For instance, the annular hole structures have been applied for detecting  biomolecular events in a label-free and highly sensitive manner from the shifts of resonant transmission peaks \cite{park15}. One fundamental question in such applications is the sensitivity analysis of resonance frequencies, where the goal is to quantify how the transmission peaks shift with respect to the refractive index change or the profile change of the biochemical samples.

\appendix
\section{Dirichlet eigenvalues and eigenfunctions}
In this section, we characterize the asymptotic behavior of the eigenpair
$\{\lambda_{mn}^D,\psi_{mn}^D\}_{m\in\mathbb{Z},n\in\mathbb{N}^*}$ of the Dirichlet eigenvalue problem (DEP)  for $0<h\ll1$.
\begin{mylemma}
  \label{lem:dir:eig}
  For $h\ll 1$ and $m\in\mathbb{N}$, the nonzero roots of 
  (\ref{eq:gov:roots}) admit the expansion
  \begin{equation}
    \label{eq:asy:betamn}
    \beta_{mn}^{D}(h)=\frac{n\pi}{h} + \frac{(4m^2-1)h}{8n\pi} + n^{-3}{\cal
      O}(h^2),\quad {m=0,1,2,\cdots,\quad n=1,2,\cdots,} 
  \end{equation}
  where the prefactor in the ${\cal O}$-notation depends on $m$ only.
\begin{proof}
  Let $C>0$ be a sufficiently large constant that is independent of $h$. Since
  $\beta$ and $\beta(1+h)$ are of the same order of magnitude when $|\beta|\geq
  C$, we apply the formulas in \cite[Sec. 10.21(x)]{nist10} to obtain the
  asymptotic formula (\ref{eq:asy:betamn}) in the region
  $(-\infty,-C]\cup[C,\infty)$. The rest is to show that there are no
  roots in $(-C,C)/\{0\}$.

  We distinguish two cases: $0<|\beta|\leq c_0$ or $|\beta|\in(c_0,C)$ for some
  sufficiently small constant $c_0\geq 0$. If $0<\beta\leq c_0$, then by the
  asymptotic behaviors of $J_m$ and $Y_m$ of small arguments \cite[10.7.3 \&
  10.7.4]{nist10}, we obtain for any fixed $h>0$,
  \begin{align*}
    F_{m}(\beta,h) \eqsim &(\beta/2)^{m}/(\Gamma(m+1))(-\pi^{-1})\Gamma(m)(\beta(1+h)/2)^{-|m|}\nonumber\\
    &-(\beta(1+h)/2)^{m}/(\Gamma(m+1))(-\pi^{-1})\Gamma(m)(\beta/2)^{-|m|}\\
    =&(-( m\pi )^{-1}) \left[  (1+h)^{-m}-(1+h)^{m}\right]>0,\quad \beta \ll 1.
  \end{align*}
  Now, suppose $|\beta|\in(c_0,C)$ so that $F$ becomes analytic at $h=0$.
  Taylor's expansion directly gives rise to
  \begin{align*}
    F_{m}(\beta,h) =& (Y_{m}(\beta)J_{m}'(\beta)-J_{m}(\beta)Y_{m}'(\beta))\beta h \\
    &+ (Y_{m}(\beta){J_{m}''}(\beta+\xi_h h)-J_{m}(\beta){Y_m''}(\beta+\xi_h h)) (\beta h)^2/2, 
  \end{align*}
  for some $\xi_h\in (0,1)$ depending on $h$. Since $Y_m(\beta)$ and $J_m(\beta)$
  are linearly independent over the interval $[c_0,C]$, the first term is in fact
  strictly nonzero, so that for $h\ll 1$, $F(h;\beta,m)\neq 0$ for any
  $|\beta|\in(c_0,C)$, which concludes the proof.
\end{proof}
\end{mylemma}

The following lemma characterizes the asymptotic behavior of
$\psi_{mn}^{D}(r,\theta;h)$ for $h\ll 1$.
\begin{mylemma}
  For $h\ll 1$, $r\in[1,1+h]$, $(i,j)\in\mathbb{Z}\times\mathbb{N}^*$, we have
  \label{lem:psimnD}
  \begin{align}
    \label{eq:psimnD}
      \psi_{mn}^{D}(r,\theta;h) =& \frac{e^{\bi m\theta}}{\sqrt{\pi rh}}\Bigg\{\sin\left[ \frac{n\pi}{h}(r-1) \right] -\frac{\gamma_m^D(r)h}{n\pi}\cos\left[ \frac{n\pi}{h}(r-1) \right]+n^{-2}{\cal O}(h^2)\Bigg\},
  \end{align}
  where the function $\gamma^D_{m}:[1,1+h]\to\mathbb{R}$ is given by
  \[
    \gamma^D_{m}(x) = \frac{(4m^2-1)}{8r}(r-1).
  \]
and the prefactors in the ${\cal O}$-notations depend on $m$ only.

  \begin{proof}
    For $h\ll1$, Lemma~\ref{lem:dir:eig} implies that $\beta_{|m|n}^{D}\gg 1$, thus by \cite[10.17.3\&10.17.4]{nist10}, we have
    \begin{align*}
      &Y_{|m|}(\beta_{|m|n}^{D})J_{|m|}(\beta_{|m|n}^{D}r)-J_{|m|}(\beta_{|m|n}^{D})Y_{|m|}(\beta_{|m| n}^{D}r)\\
      =&\frac{2}{\pi\beta_{|m|n}^{D}\sqrt{r}}\Bigg\{\sin(\frac{n\pi}{h}(1-r))(1+ {\cal O}(n^{-2}h^2))+ \frac{\gamma_m^{D}(r) h}{n\pi }\cos(\frac{n\pi}{h}(1-r))(1+ {\cal O}(n^{-2}h^2))  \Bigg\}
    \end{align*}
    where the prefactor in the ${\cal O}$-notation depends only on $m$.
    Therefore, it can be verified that
    \[
      C_{mn}^{D}= \frac{2h^{1/2}}{\beta_{|m|n}^{D}\pi^{1/2}}\left[1+n^{-2}{\cal
          O}(h^2) \right],
    \]
    and hence (\ref{eq:psimnD}) follows.
  \end{proof}
\end{mylemma}

\section{Neumann eigenvalues and eigenfunctions}
In this section, we derive the asymptotic expansion of the eigenpair
$\{\lambda_{mn}^N,\psi_{mn}^N\}_{m\in\mathbb{Z},n\in\mathbb{N}}$ for the Neumann eigenvalue problem (NEP) when $0<h\ll1$.
\begin{mylemma}
\label{lem:neu:eig}
For $h\ll 1$, the nonzero roots to
(\ref{eq:gov:roots:n}) with sufficiently large magnitude attain the expansion
  \begin{equation}
    \label{eq:asy:betamn:N}
    \beta_{mn}^{N}=\frac{n\pi}{h} + \frac{(4m^2+3)h}{8n\pi(1+h)} + n^{-3}{\cal
      O}(h^3),\quad {m=0,1,2,\cdots,\quad n=1,2,\cdots,} 
  \end{equation}
  where the prefactor in the ${\cal O}$-notation depends on $m$. On
    the other hand, when $m\neq 0$, there exists a unique root close to $m$ satisfying
  \begin{equation}
    \label{eq:asy:betam2:N}
    \beta_{m0}^{N}=m -\frac{mh}{2} + {\cal O}(h^2),\quad {m=1,2,\cdots.} 
  \end{equation}
\begin{proof}
  When $m=0$,
  \begin{align*}
    F_0^{N}(\beta;h)=Y_{1}(\beta)J_{1}\left( \beta(1+h) \right) - J_{1}(\beta)Y_{1}\left( \beta(1+h) \right) = F_1^{D}(\beta,h),
  \end{align*}
  so that Lemma~\ref{lem:dir:eig} applies. In the following, we assume $m\neq
  0$.

  Let $C>0$ be a sufficiently large constant that is independent of $h$. Since
  $\beta$ and $\beta(1+h)$ are of the same order of magnitude when $\beta\geq
  C$, the asymptotic formula (\ref{eq:asy:betamn}) in the region $[C,\infty)$ can be derived using the formulas in \cite[Sec. 10.21(x)]{nist10}. We only need to show that there is only one root in $(0,C)$ and it is near $m$.

  We distinguish two cases: $0<\beta\leq c_0$ or $\beta\in(c_0,C)$ for some
  sufficiently small constant $c_0\geq 0$ independent of $h$. If $\beta\leq
  c_0$, by the power series representation of $J_m$ and $Y_m$ with small arguments \cite[10.2.2 \&
  10.8.1]{nist10}, we obtain for any fixed $h>0$,
  \begin{align*}
    F_m^{N}(\beta;h) \eqsim & \frac{m!2^m}{\pi\beta^{m+1}}\frac{\left( \beta(1+h)\right)^{m-1}}{2^m (m-1)!} - \frac{m!2^m}{\pi\beta^{m+1}(1+h)^{m+1}}\frac{\left( \beta\right)^{m-1}}{2^m (m-1)!}    \\
      \eqsim & \frac{m}{\pi\beta^2}\left[(1+h)^{m-1}-(1+h)^{-m-1}  \right]>0,\quad \beta \ll 1.
  \end{align*}
  Now if $\beta\in(c_0,C)$ so that $F_m^{N}$ is analytic at
  $h=0$. Taylor's expansion of $F_m^{N}$ at $h=0$ and Bessel's differential
  equation directly give rise to
  \begin{align*}
    F_m^{N}(\beta;h) =& (Y_m'(\beta)J_m{\dd}(\beta)-J_m'(\beta)Y_m{\dd}(\beta))\beta h + {\cal O}(h^2)\\
    =& -(1-m^2\beta^{-2})\left[Y_m'(\beta)J_m(\beta)-J_m'(\beta)Y_m(\beta)\right]\beta h+ {\cal O}(h^2).
  \end{align*}
  Since $Y_m(\beta)$ and $J_m(\beta)$ are linearly independent over the interval
  $[c_0,C]$, the first term is in fact strictly nonzero and is far greater than
  the second term if $(\beta - m)\gg h $ for $h\ll 1$. But the intermediate value
  theorem implies that there exists a root in $(m-c_0,m+c_0)$ satisfying
  $\beta-m={\cal O}(h)$. By a similar asymptotic analysis of
  $F_m^{N}(\beta;h)$ at $\beta=m$, the expansion (\ref{eq:asy:betam2:N}) follows.
\end{proof}
\end{mylemma}
The following lemma characterizes the asymptotic behavior of
$\psi_{mn}^{N}(r,\theta;h)$ for $h\ll 1$.
\begin{mylemma}
  \label{lem:psimnN}
  Let $h\ll 1$, $r\in[1,1+h]$, $m\in\mathbb{Z}$ and $n\in\mathbb{N}$. If $n\geq 1$, then 
  \begin{align}
    \label{eq:psimnN:n1}
    \psi_{mn}^{N}(r,\theta;h) =&\frac{e^{\bi m\theta}}{\sqrt{\pi rh}}\Bigg\{\cos\left[ \frac{n\pi}{h}(r-1) \right] +\frac{\gamma^N_{m}(r) h}{n\pi}\sin\left[ \frac{n\pi}{h}(r-1) \right]+{\cal O}(n^{-2}h^2)\Bigg\},
  \end{align}
  where the function $\gamma^N_{m}:[1,1+h]\to\mathbb{R}$ is given by
  \[
    \gamma^N_{m}(r) = \frac{(4m^2+3)h}{8(1+h)} - \frac{(4m^2-1)}{8 r}. 
  \]
  If $n=0$,
  \begin{equation}
    \label{eq:psimnN:n0}
     \psi_{m0}^{N}(r,\theta;h) = \frac{e^{\bi m\theta}}{\sqrt{\pi h(h+2)}}  + {\cal O}(h^{3/2}).
  \end{equation}
In the above, the prefactors in the ${\cal O}$-notations depend on $m$ only.

  \begin{proof}
        For $n\ge1$, it follows from Lemma~\ref{lem:neu:eig} that $\beta_{|m|n}^{N}\gg
    1$. Thus by \cite[10.17:3,4,9,10]{nist10}, we obtain
    \begin{align*}
      &Y_{|m|}'(\beta_{|m|n}^{N})J_{|m|}(\beta_{|m|n}^{N}r)-J_{|m|}'(\beta_{|m|n}^{N})Y_{|m|}(\beta_{|m|n}^{N}r)\\
      =&\frac{2}{\pi\beta_{|m|n}^{N}}\sqrt{\frac{1}{r}}\left\{ \cos\left[ \frac{n\pi}{h}(r-1) \right](1+{\cal O}(n^{-2}h^2)) +\frac{\gamma_{m}(r;h)h}{n\pi}\sin\left[\frac{n\pi}{h}(r-1) \right](1+{\cal O}(n^{-2}h^2))\right\},
    \end{align*}
    where the prefactor in the ${\cal O}$-notation depends only on $m$. It
    can be verified that the normalization constant
    \begin{align*}
      C_{mn}^{N} 
      =&\frac{2\sqrt{\pi h} }{\pi\beta_{|m|n}^{N}}(1 + {\cal O}(n^{-2}h^2)),
    \end{align*}
    so that (\ref{eq:psimnN:n1}) follows. 
    If $n=0$, we have
    \begin{align*}
      &Y_{|m|}'(\beta_{|m|0}^{N})J_{|m|}(\beta_{|m|0}^{N}r)-J_{|m|}'(\beta_{|m|0}^{N})Y_{|m|}(\beta_{|m|0}^{N}r)\\
      =&\left[ Y_{|m|}'(\beta_{|m|0}^{N})J_{|m|}(\beta_{|m|0}^{N})-J_{|m|}'(\beta_{|m|0}^{N})Y_{|m|}(\beta_{|m|0}^{N})\right]\left[1 + {\cal O}(h^2)\right].
    \end{align*}
    Then
    \begin{align*}
      C_{m0}^{N} = \sqrt{\pi h(h+2)}\left[ Y_{|m|}'(\beta_{|m|0}^{N})J_{|m|}(\beta_{|m|0}^{N})-J_{|m|}'(\beta_{|m|0}^{N})Y_{|m|}(\beta_{|m|0}^{N})\right]\left[1 + {\cal O}(h^2)\right],
    \end{align*}
    and hence (\ref{eq:psimnN:n0}) follows.
  \end{proof}
\end{mylemma}

\section{Asymptotic analysis of matrix elements in \eqref{eq:eig:m:2}}
To prove Lemma~\ref{lem:dm'm:h}, we need the following two lemmas.
\begin{mylemma}
  \label{prop:S0}
  For any two functions $f,g\in C_{\rm comp}^{\infty}(\mathbb{R})$ with
  \[
\max\{||f||_{W_{\infty}^1(\mathbb{R})},||g||_{W_{\infty}^1(\mathbb{R})}\}\leq M,
\]
for some constant $M>0$, the following three INF matrices
  \begin{align*}
    &\{({\cal S}_0[( n' )^{-1/2}f(\cdot)\sin(n'\pi\cdot)],n^{-1/2}g(\cdot)\sin(n\pi\cdot))_{L^2(0,1)}\}_{n,n'=1}^{\infty},\\
    &\{({\cal S}_0[( n' )^{1/2}f(\cdot)\cos(n'\pi\cdot)],n^{-1/2}g(\cdot)\sin(n\pi\cdot))_{L^2(0,1)}\}_{n,n'=1}^{\infty},\\
    &\{({\cal S}_0[( n' )^{1/2}f(\cdot)\cos(n'\pi\cdot)],n^{1/2}g(\cdot)\cos(n\pi\cdot))_{L^2(0,1)}\}_{n,n'=1}^{\infty},
  \end{align*}
  are bounded operators mapping from $\ell^2$ to $\ell^2$, with norms
  depending only on $M$ but not on functions $f$ and $g$.
  \begin{proof}
    We only give the proof for the first matrix.
    For any $n\in\mathbb{N}$,
    \begin{align*}
      ((n')^{-1/2}\sin(n'\pi\cdot),\phi_n)_{L^2(0,1)} =&
                                                           \begin{cases}
                                                             \frac{1}{\sqrt{n'}}\frac{1-(-1)^{n'}}{n'\pi}& n = 0,\\
                                                             \frac{1}{2\sqrt{2n'}}\left[ \frac{1-(-1)^{(n'+n)}}{(n'+n)\pi} + \frac{1-(-1)^{(n'-n)}}{(n'-n)\pi} \right] & n\neq n',\\
                                                             \frac{1}{2\sqrt{2n'}} \frac{1-(-1)^{(n'+n)}}{2n'\pi} & n = n'.
                                                           \end{cases}
    \end{align*}
    Thus for any $\{a_{n'}\}_{n'=1}^{\infty}\in\ell^2$ and any $N,N'\in\mathbb{N}^*$, we have
    \begin{align*}
      &\left|\sum_{n'=1}^{N'}a_{n'}((n')^{-1/2}\sin(n'\pi\cdot),\phi_n)_{L^2(0,1)} \right|\\
      \leq &\left(\sum_{n'=1}^{\infty}|a_{n'}|^2  \right)^{1/2}\left(\sum_{n'=1}^{\infty} |((n')^{-1/2}\sin(n'\pi\cdot),\phi_n)_{L^2(0,1)}|^2\right)^{1/2} <\infty
    \end{align*}
    for $n\leq N$. We obtain
    \begin{align*}
      &\sum_{n=0}^{N}(1+n^2)^{-1/2}\left|\sum_{n'=1}^{N'}(a_{n'}(n')^{-1/2}\sin(n'\pi\cdot),\phi_n)_{L^2(0,1)}\right|^2\\
      \leq& \left(\sum_{n'=1}^{\infty}|a_{n'}|^2  \right)\sum_{n=0}^{N}\sum_{n'=1}^{N'}(1+n^2)^{-1/2}\left|((n')^{-1/2}\sin(n'\pi\cdot),\phi_n)_{L^2(0,1)}\right|^2\\
      \leq& \left(\sum_{n'=1}^{\infty}|a_{n'}|^2  \right)\Bigg[\sum_{n=0}^{\infty}(1+n^2)^{-1/2}\frac{C}{n^3} +\sum_{0\le n\le N, 1\le \delta n'\le N'}(1+n^2)^{-1/2}\frac{C}{\delta n'^2(\delta n'+n)} \\
    &+\sum_{1\le \delta n\le M', 1\le n'\le N'}(1+(\delta n+n')^2)^{-1/2}\frac{C}{\delta n^2n'}\Bigg]\leq C ||\{a_{n'}\}||_{\ell^2}^2,
    \end{align*}
    where $C>0$ denotes a generic and sufficiently large constant.
    The above implies that
    \[
      \phi(r'):=\sum_{n'=1}^{\infty}a_{n'}(n')^{-1/2}\sin(n'\pi r')\in
      \widetilde{H^{-1/2}}(0,1),
    \]
    with $||\phi||_{\widetilde{H^{-1/2}}(0,1)}\le C||\{a_n\}||_{\ell^2}$.
    Similarly, for any $\{b_n\}_{n=1}^{\infty}\in\ell^2$,
    \[
      \psi(r'):=\sum_{n=1}^{\infty}b_{n}(n)^{-1/2}\sin(n\pi r)\in \widetilde{H^{-1/2}}(0,1),
    \]
    with $||\psi||_{\widetilde{H^{-1/2}}(0,1)}\le C||\{b_n\}||_{\ell^2}$. Therefore,
    \begin{align*}
      &|\sum_{n=1}^{\infty}\sum_{n'=1}^{\infty}a_{n'}({\cal S}_0[( n' )^{-1/2}f(\cdot)\sin(n'\pi\cdot)],n^{-1/2}g(\cdot)\sin(n\pi\cdot))_{L^2(0,1)} b_{n}|\\
      =&|\langle {\cal S}_0 f\phi,g\psi \rangle_{(0,1)}| \leq C||f\phi||_{\widetilde{H}^{-1/2}(0,1)}||g\psi||_{\widetilde{H}^{-1/2}(0,1)} \leq C(M)||\{a_{n'}\}_{n'=1}^{\infty}||_{\ell^2}||\{b_{n}\}_{n=1}^{\infty}||_{\ell^2},
    \end{align*}
    where $C(M)$ denotes the dependence of the constant $C$ on $M$ \cite[Thm. 3.20]{mcl00}, 
    implying the boundedness of the first matrix mapping from $\ell^2$ to
    $\ell^2$.
  \end{proof}
\end{mylemma}

Define
\begin{align}
    F_m(r,r'):=\frac{1}{2\pi}\int_{0}^{2\pi} \frac{e^{\bi k\sqrt{r^2+r'^2-2rr'\cos\theta}}}{\sqrt{r^2+r'^2-2rr'\cos\theta}}e^{\bi m\theta}d\theta.
  \end{align}
\begin{mylemma}
  \label{prop:FourierG}
  For $k\in{\cal B}$ and $r\neq r'\in(1,1+h)$, we have
  \begin{align}
    F_m(r,r')=& \frac{1}{\sqrt{rr'}}\Bigg[\log(1-w^{-2})\left[\frac{(-1)}{4\pi} + (1-w^{-2})f_m(1-w^{-2},\lambda^2)  \right]\nonumber\\
    &+ g_m(1-w^{-2},\lambda^2)  \Bigg],
  \end{align}
  where $w=(r^2+r'^2)/(2rr')$, $\lambda=k\sqrt{2rr'}$, and $f_m(t_1,t_2)$ and
  $g_m(t_1,t_2)$ are analytic for $t_1\in (-1,1)$ and $|t_2|<C$ for some
  sufficiently large constant $C$.
  \begin{proof}
    By Eqs.~(45-48) in \cite{concoh10}, we have
    \begin{align*}
      F_m(r,r') =& \frac{1}{2}\left[\hat{\Lambda}_+^m(rr',\lambda,w) + \hat{\Lambda}_-^m(rr',\lambda,w)  \right],
    \end{align*}
    where 
    \begin{align*}
      \hat{\Lambda}_+^m(rr',\lambda,w) =& \frac{(-1)^m}{\sqrt{rr'}}\sum_{p=0}^{\infty}\left( \frac{\lambda^2\sqrt{w^2-1}}{4} \right)^{p}\frac{Q_{m-1/2}^p(w)}{p!\Gamma(p-m+1/2)\Gamma(p+m+1/2)},\\
      \hat{\Lambda}_-^m(rr',\lambda,w) =& \frac{(-1)^m}{\sqrt{rr'}}\sum_{p=0}^{\infty}\left( \frac{\lambda^2\sqrt{w^2-1}}{4} \right)^{p+m+1/2}\frac{Q_{m-1/2}^{p+m+1/2}(w)}{p!\Gamma(p+m+3/2)\Gamma(p+2m+1)},
    \end{align*}
    and $Q_v^\mu$ denotes the associated Legendre function \cite[\S 14.1]{nist10}.
    By Eqs.~(14.3.7), (15.8.10), and (15.8.12) in~\cite{nist10}, we obtain
    \begin{align*}
      \hat{\Lambda}_+^m(rr',\lambda,w) =&\frac{(-1)^m}{\sqrt{rr'}}\sum_{p=0}^{\infty}\left( \frac{\lambda^2(w^2-1)}{4} \right)^{p}\frac{\sqrt{\pi}e^{\bi p\pi} {}_2F_1(\frac{m+p+1/2}{2},\frac{m+p+3/2}{2};m+1;w^{-2})}{2^{m+1/2}p!w^{m+p+1/2}\Gamma(p-m+1/2)}\\
      =&\frac{(-1)^m}{\sqrt{rr'}}\sum_{p=0}^{\infty}\left( \frac{\lambda^2}{4w^2} \right)^{p}\frac{\sqrt{\pi}e^{\bi p\pi} {}_2F_1(\frac{m-p+1/2}{2},\frac{m-p+3/2}{2};m+1;w^{-2})}{2^{m+1/2}p!w^{m+p+1/2}\Gamma(p-m+1/2)}\\
      =&\frac{1}{\sqrt{rr'}}\left[\log(1-w^{-2})\frac{(-1)}{2\pi} + \log(1-w^{-2})(1-w^{-2})f^1_m(1-w^{-2},\lambda^2) + g^1_m(1-w^{-2},\lambda^2)  \right]
    \end{align*}
    for some analytic functions $f^1_m(t_1,t_2)$ and $g^1_m(t_1,t_2)$, where
    ${}_2F_1$ denotes the hypergeometric function \cite[\S 15.1]{nist10}. 
    Let $z=\frac{1}{2}\left( 1-\frac{1}{1-w^{-2}}\right)$. By Eqs.~(14.3.7), (15.8.19) and (15.8.8) in \cite{nist10}, we have
    \begin{align*}
      \hat{\Lambda}_-^m(rr',\lambda,w) =& \frac{1}{\sqrt{rr'}}\sum_{p=0}^{\infty}\left( \frac{\lambda^2(w^2-1)}{4} \right)^{p+m+1/2}\frac{\sqrt{\pi} e^{\bi (p+1/2)\pi}{}_2F_1(\frac{p+2m+1}{2},\frac{p+2m+2}{2};m+1;w^{-2})}{p!w^{p+2m+1}\Gamma(p+m+3/2)}\\
      =&\frac{1}{\sqrt{rr'}}\sum_{p=0}^{\infty}\left( \frac{\lambda^2(w^2-1)}{4} \right)^{p+m+1/2}(1-2z)^{p+2m+1}(1-z)^{-m}\sqrt{\pi} e^{\bi (p+1/2)\pi}\\
      &\frac{{}_2F_1(-p-m,p+m+1;m+1;z)}{p!w^{p+2m+1}\Gamma(p+m+3/2)},\\
      =&\frac{1}{\sqrt{rr'}}\sum_{p=0}^{\infty}\left( \frac{\lambda^2}{4} \right)^{p+m+1/2}(1-z)^{-m}\sqrt{\pi} e^{\bi (p+1/2)\pi}\frac{{}_2F_1(-p-m,p+m+1;m+1;z)}{p!(4z(z-1))^{p/2}\Gamma(p+m+3/2)}\\
      =&\frac{1}{\sqrt{rr'}}\left[\log(1-w^{-2})(1-w^{-2})f^2_m(w^{-2}-1,\lambda^2) + g^2_m(w^{-2}-1,\lambda^2)  \right],
    \end{align*}
    for some analytic functions $f^2_m(t_1,t_2)$ and $g^2_m(t_1,t_2)$. The proof is complete by
    taking $f_m=f^1_m+f^2_m$ and $g_m=g^1_m+g^2_m$.
  \end{proof}
\end{mylemma}

\noindent{\bf Proof of Lemma~\ref{lem:dm'm:h}:}
(i). It follows from \eqref{eq:LkSk} that
  \begin{equation*}
      B^{TE,TE}_{n'n;m} = 2\frac{(1 + {\cal O}[(\lambda_{mn'}^N)^{-1}])}{(\lambda^N_{mn'}\lambda_{mn}^N)^{3/4}}\left[k^2\langle {\cal S}_k \nabla_2\psi_{mn}^N,\nabla_2\bar{\psi_{mn}^N} \rangle_{A^h} - \lambda_{mn}^N\lambda_{mn'}^N \langle {\cal S}_k\psi_{mn}^N,\bar{\psi_{mn}^N} \rangle_{A^h} \right].
  \end{equation*}
  By Lemmas~\ref{lem:psimnN} and \ref{prop:FourierG}, we obtain
    \begin{align*}
&[\lambda_{mn}^N\lambda_{mn'}^N]^{1/4} \langle {\cal S}_k\psi_{mn}^N,\bar{\psi_{mn}^N} \rangle_{A^h}\\ 
      =&(nn')^{1/2}\pi(1+{\cal O}(h^2))\int_{0}^{1}dr\int_{0}^{1}F_m(1+hr,1+hr')\sqrt{1+hr}\sqrt{1+hr'}\\ 
                          &\left[\cos(n'\pi r') +\frac{\gamma_{m}(1+hr';h) h}{n'\pi}\sin(n'\pi r')+{\cal O}((n')^{-2}h^2)\right]\\
                          &\left[\cos(n\pi r) +\frac{\gamma_{m}(1+hr;h) h}{n\pi}\sin(n\pi r)+{\cal O}(n^{-2}h^2)\right]dr',
    \end{align*}
    where the prefactors in the ${\cal O}$-notations do not
    depend on $n,n'$. Using Lemma  \ref{prop:FourierG}, it follows that
    \[
      F_m(1+hr,1+hr') = -\frac{\log [h|r-r'|]}{2\pi\sqrt{(1+hr)(1+hr')}} +
      2h^2\log [h|r-r'|]\tilde{f}_m(1+hr,1+hr')  + \tilde{g}_m(1+hr,1+hr'),
    \]
    for some analytic functions $\tilde{f}_m$ and $\tilde{g}_m$. It is
    straightforward to verify using Lemma~\ref{prop:S0} that
    \[
[\lambda_{mn}^N\lambda_{mn'}^N]^{1/4} \langle {\cal
  S}_k\psi_{mn}^N,\bar{\psi_{mn}^N} \rangle_{A^h} = ({\cal S}_0[(n'\pi)^{1/2}\phi_{n'}],(n\pi)^{1/2}\phi_n)_{L^2(0,1)} + {\cal O}(h)\epsilon_{n'n;m}^{TE,TE;1},
    \]
    where the INF matrix $[\epsilon_{n'n;m}^{TE,TE;1}]$ is uniformly bounded for $h\ll1$. One similarly verifies that 
    \[
[\lambda_{mn}^N\lambda_{mn'}^N]^{-3/4} k^2\langle {\cal
  S}_k\nabla_2\psi_{mn}^N,\nabla_2\bar{\psi_{mn}^N} \rangle_{A^h} = {\cal O}(h)\epsilon_{n'n;m}^{TE,TE;2},
    \]
    where the INF matrix $[\epsilon_{n'n;m}^{TE,TE;2}]$ is uniformly bounded $h\ll1$, and (\ref{eq:n'n:TETE}) follows immediately. \\

    \medskip
    \noindent(ii-vi). The proof follows from similar arguments as in (i). We omit the details
    here. \\

    \medskip
    
    \noindent(vii). According to \eqref{eq:LkSk}, 
    \begin{align*}
      A_{mm} =& \, 2\cos(s_{m0}^Nl/2)\lambda_{m0}^N \langle S_k\psi_{m0}^N,\overline{\psi_{m0}^N} \rangle_{A^h}-\frac{2k^2\cos(s_{m0}^Nl/2)}{\lambda_{m0}}\langle S_k\nabla_2\psi_{m0}^N,\nabla_2\overline{\psi_{m0}^N} \rangle_{A^h}.
    \end{align*}
    Similar to the derivations in (i), Lemma~\ref{prop:FourierG} implies that
    \begin{align}
      \langle S_k\psi_{m0}^N,\overline{\psi_{m0}^N} \rangle_{A^h}=& \,      h\int_{0}^{1}\int_{0}^{1}F_m(1+hr,1+hr')\frac{(1+hr)(1+hr')}{2+h}(1 +{\cal O}(h^2))drdr' \nonumber \\
      =&-h\frac{\log h}{4\pi} + \alpha_m(k)h + \bi \beta_m(k)h + {\cal O}(h^2\log h) \label{eq:S_psi0_psi}
    \end{align}
    for some constants $\alpha_m(k)$ and $\beta_m(k)$, both of which are
    analytic in ${\cal B}$. By (50) in \cite{concoh10} and (14.8.9) in \cite{nist10},
\begin{align*}
  F_m(1+hr,1+hr') =& \frac{1}{\pi}\int_{0}^{\pi}\frac{(e^{\bi k\sqrt{h^2(r-r')^2 + 4(1+hr)(1+hr')\sin^2(\theta/2)}}-1)\cos(m\theta)}{\sqrt{h^2(r-r')^2 + 4(1+hr)(1+hr')\sin^2(\theta/2)}}d\theta \\
                   &+ \frac{Q_{m-1/2}\left[ [(1+hr)^2+(1+hr')^2]/(2(1+hr)(1+hr')) \right]}{2\pi\sqrt{(1+hr)(1+hr')}} \\
  =&\frac{1}{\pi}\int_{0}^{\pi}\frac{(e^{\bi k\sin(\theta/2)}-1)\cos(m\theta)}{\sin(\theta/2)}d\theta -\frac{1}{2\pi}\log [h|r-r'|] \\
  &+ \frac{1}{2\pi}[\log 2 -\gamma - \psi(m+1/2)] + o(1).
\end{align*}
Plugging the above into \eqref{eq:S_psi0_psi} and comparing both sides lead to (\ref{eq:alpham:k}) and (\ref{eq:betam:k}); the second
equality in \eqref{eq:betam:k} follows from \cite[Eq. (10.9.2)]{nist10}. Similarly one can show that
\begin{align*}
\langle S_k\nabla_2\psi_{m0}^N,\nabla_2\overline{\psi_{m0}^N} \rangle_{A^h} = m^2[-h\frac{\log h}{4\pi} + \tilde{\alpha}_m(k)h + \bi \tilde{\beta}_m(k)h] + {\cal O}(h^2\log h).
\end{align*}
Eq.~\eqref{eq:Amm:m=0} for $m=0$ can be verified similarly.

\bibliographystyle{plain}
\bibliography{wt}
\end{document}